\documentclass[sigconf,screen]{acmart}

\usepackage[T1]{fontenc}
\usepackage[utf8]{inputenc}
\usepackage{booktabs}
\usepackage{multirow}
\usepackage{balance}
\usepackage{xcolor} %

\usepackage{verbatim} %
\usepackage{graphicx}
\usepackage{caption}
\usepackage{subcaption} %
\usepackage{setspace} %
\usepackage{float}

\usepackage{tikz}
\usepackage{pgfplots,pgfplotstable} %
\usetikzlibrary{patterns}
\usetikzlibrary{matrix,decorations.pathreplacing, calc, positioning}
\usetikzlibrary{graphs,graphs.standard,quotes}
\usetikzlibrary{arrows, arrows.meta, automata}
\usepgfplotslibrary{groupplots}

\pgfplotsset{compat=1.16,
  /pgfplots/ybar legend/.style={
  /pgfplots/legend image code/.code={%
      \draw[##1,/tikz/.cd,yshift=-0.2em]
      (0cm,0cm) rectangle (10pt,4pt);},
  },
}
\definecolor{ctorange}{RGB}{213, 94, 0}
\definecolor{ctblue}{RGB}{0, 114, 178}
\definecolor{ctyellow}{RGB}{240, 228, 66}
\definecolor{ctgreen}{RGB}{0, 158, 115}
\definecolor{ctdarkyellow}{RGB}{230, 159, 0}
\pgfplotscreateplotcyclelist{three-colors}{
    {draw=none,fill=ctorange},
    {draw=none,fill=ctyellow},
    {draw=none,fill=ctblue},
}
\pgfplotscreateplotcyclelist{noadapt-adapt}{
    {draw=none,fill=ctorange},
    {draw=none,fill=ctyellow},
}
\pgfplotscreateplotcyclelist{gs-gb}{
    {draw=none,fill=ctdarkyellow},
    {draw=none,fill=ctgreen},
}

\usepackage[ruled,vlined]{algorithm2e}
\usepackage{listings}
\definecolor{key-color}{rgb}{0.8, 0.47, 0.196}
\usepackage{enumitem} %
\usepackage{hyphenat} %

\usepackage{doi} %
\usepackage{hyperref}
\expandafter\def\expandafter\UrlBreaks\expandafter{\UrlBreaks\do\/\do\-\do\.} %
\definecolor[named]{Purple}{cmyk}{0.55,1,0,0.15}
\definecolor[named]{DarkBlue}{cmyk}{1,0.58,0,0.21}
\hypersetup{bookmarksnumbered,unicode,naturalnames,colorlinks,breaklinks,
linkcolor=Purple,
citecolor=Purple,
urlcolor=DarkBlue,
filecolor=DarkBlue}

\usepackage{array}
\newcolumntype{L}[1]{>{\raggedright\let\newline\\\arraybackslash\hspace{0pt}}m{#1}}

\newenvironment{squishedlist}
{
  \begin{list}{$\bullet$}
   {
     \setlength{\itemsep}{0pt}
     \setlength{\parsep}{2pt}
     \setlength{\topsep}{1.0pt}
     \setlength{\partopsep}{0pt}
     \setlength{\leftmargin}{01.5em}
     \setlength{\labelwidth}{1em}
     \setlength{\labelsep}{0.5em}
   }
}
{
   \end{list}
}
\newenvironment{squishedenumerate}
{
  \begin{enumerate}[label=(\roman*),topsep=0.2em,itemsep=0.2em,leftmargin=1.7em] %
}
{
  \end{enumerate}
}

\newcommand{\noun}[1]{\textsc{#1}}
\newcommand{\graphsurge}[1]{\noun{Graphsurge}}
\newcommand{\customsection}[1]{\noindent\textbf{#1}}
\newcommand{\diffonly}[1]{\texttt{diff-only}}
\newcommand{\scratch}[1]{\texttt{scratch}}
\newcommand{\adaptive}[1]{\texttt{adaptive}}

\newlength{\subsubheadingvspace}
\newlength{\imagebottommargin}
\newtheorem{theorem}{Theorem}[section]

\newcounter{theorem2}
\newtheorem{example}[theorem2]{Example}

\newcounter{theorem3}
\newtheorem{definition}[theorem3]{Definition}

\copyrightyear{2021}
\acmYear{2021}
\setcopyright{acmlicensed}
\acmConference[SIGMOD '21] {Proceedings of the 2021 International Conference on Management of Data}{June 18--27, 2021}{Virtual Event, China}
\acmBooktitle{Proceedings of the 2021 International Conference on Management of Data (SIGMOD '21), June 18--27, 2021, Virtual Event, China}
\acmPrice{15.00}
\acmDOI{10.1145/3448016.3452837}
\acmISBN{978-1-4503-8343-1/21/06}

\begin{document}

\fancyhead{}

\title{\texorpdfstring{Graphsurge: Graph Analytics on View Collections Using Differential Computation}{Graphsurge}} %

\author{Siddhartha Sahu}
\email{s3sahu@uwaterloo.ca}
\affiliation{%
  \institution{University of Waterloo, Canada}
  \city{Waterloo}
  \state{ON}
  \country{Canada}
}

\author{Semih Salihoglu}
\email{semih.salihoglu@uwaterloo.ca}
\affiliation{%
  \institution{University of Waterloo, Canada}
  \city{Waterloo}
  \state{ON}
  \country{Canada}
}

\begin{abstract}
This paper presents the design and implementation of a new open-source view-based graph analytics system called \graphsurge{}. \graphsurge{} is designed to support applications that analyze multiple snapshots or views of a large-scale graph. Users program \graphsurge{} through a declarative \emph{graph view definition language} (GVDL) to create views over input graphs and a \emph{Differential Dataflow}-based programming API to write analytics computations.
A key feature of GVDL is the ability to organize views into \emph{view collections}, which allows \graphsurge{} to automatically share computation across views, without users writing any incrementalization  code, by performing computations differentially. We then introduce two optimization problems that naturally arise in our setting. First is the \emph{collection ordering problem} to determine the order of views that leads to minimum differences across consecutive views. We prove this problem is NP-hard and show a constant-factor approximation algorithm drawn from literature. Second is the \emph{collection splitting} problem to decide on which views to run computations differentially vs from scratch, for which we present an adaptive solution that makes decisions at runtime. We present extensive experiments to demonstrate the benefits of running computations differentially for view collections and our collection ordering and splitting optimizations.

\vspace{-5pt}

\end{abstract}



\maketitle

\section{Introduction}%
\label{sec:introduction}

A variety of applications, such as fraud detection, risk assessment
and recommendations from telecommunications,
finance, social networking, biological brain networks, and many other fields, process large-scale connected data among
different entities~\cite{sahu:extended-survey}. Developers of these
applications naturally model such connected data as graphs. Many of these
applications require the ability to analyze different snapshots or \emph{views}
of a large-scale static graph, often based on selecting subsets of nodes or edges that satisfy different predicates. We first review several of these applications that motivate
our current work. Figure~\ref{fig:call-graph-network} shows a call graph that we use as a running example throughout the paper. Customers are represented as nodes with \texttt{profession} and \texttt{city} properties. Phone calls are represented as edges between nodes with \texttt{duration} and \texttt{date} properties, written in curly brackets, respectively.

\begin{figure}[t!]
	\centering
	\captionsetup{justification=centering}
\centering
\begin{tikzpicture}[scale=.6, transform shape,->,>=stealth', shorten >=1pt, auto,node distance=3.5cm, thick, main node/.style={circle,draw,font=\Large\bfseries,minimum size=1cm, fill=cyan}, text node/.style={rectangle,draw}]
	\node[main node] (5) {$5$};
	\node[text node] (text5)[node distance = 2.5mm, yshift=-1.2cm, xshift=0.8cm, rotate=-60, right of=5] {{NY, Doctor}};
	\node[main node] (8) [left of=5] {$8$};
	\node[text node] (text8)[node distance=10mm, above of=8] {LA, Lawyer};
	\node[main node] (2) [right of=5] {$2$};
	\node[text node] (text2)[node distance =5mm, yshift=-1cm, xshift=-.6cm, rotate=60, left of=2] {{LA, Doctor}};
	\node[main node] (6) [above of=5] {$6$};
	\node[text node] (text6)[node distance=10mm, above of=6] {LA, Engineer};
	\node[main node] (1) [above of=2] {$1$};
	\node[text node] (text1)[node distance=10mm, above of=1] {LA, Engineer};
	\node[main node] (3) [below of=2] {$3$};
	\node[text node] (text3)[node distance=10mm, below of=3] {LA, Engineer};
	\node[main node] (7) [below of=8] {$7$};
	\node[text node] (text7)[node distance=10mm, below of=7] {NY, Lawyer};
	\node[main node] (4) [below of=5] {$4$};
	\node[text node] (text4)[node distance=10mm, below of=4] {NY, Lawyer};
	\path [every node/.style={font=\sffamily\large}]
	([yshift=-1.5mm]5.east) edge node[midway, below, sloped]{\{$7, 2015$\}}  ([yshift=-1.5mm]2.west)
	([yshift=1.5mm]2.west) edge node[midway, above, sloped]{\{$19, 2019$\}}  ([yshift=1.5mm]5.east)
	([xshift=1.5mm]8.south) edge node[midway, below, sloped]{\{$13, 2019$\}}  ([xshift=1.5mm]7.north)
	([xshift=-1.5mm]8.south) edge node[midway, above, sloped]{\{$18, 2019$\}}  ([xshift=-1.5mm]7.north)
	(8) edge node[midway, above, sloped]{\{$6, 2019$\}}  (5)
	(5) edge node[midway, below, sloped]{\{$18, 2019$\}}  (4)
	(4) edge node[midway, below, sloped]{\{$32, 2017$\}}  (3)
	(2) edge node[midway, above, sloped]{\{$1, 2010$\}}  (3)
	(1) edge node[midway, below, sloped]{\{$10, 2018$\}}  (5)
	([xshift=1.5mm]1.south) edge node[midway, below, sloped]{\{$3, 2019$\}}  ([xshift=1.5mm]2.north)
	([xshift=-1.5mm]1.south) edge node[midway, above, sloped]{\{$12, 2017$\}}  ([xshift=-1.5mm]2.north)
	([xshift=-1.5mm]6.south) edge node[midway, below, sloped]{\{$7, 2018$\}}  ([xshift=-1.5mm]5.north)
	([xshift=1.5mm]5.north) edge node[midway, below, sloped]{\{$2, 2013$\}}  ([xshift=1.5mm]6.south)
	(6) edge node[midway, below, sloped]{\{$4, 2019$\}}  (1)
	(5) edge node[midway, below, sloped]{\{$34, 2019$\}}  (7);
\end{tikzpicture}
    \caption{Example phone call graph.}%
    \label{fig:call-graph-network}
    \vspace{-17pt}
\end{figure}
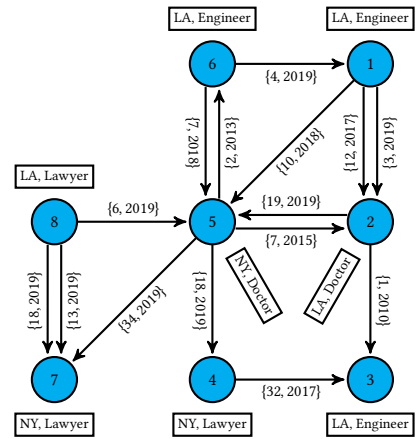

\begin{example}\label{ex:network-analysis-over-evolving-graphs}
\emph{
Researchers and practitioners study the changes in structural properties of graphs across different
views.
A popular example is historical analyses of graphs where nodes or edges have some time property. A network scientist might study the history of
the connectivity of the call graph from Figure~\ref{fig:call-graph-network} and compute one view of the graph for each year between 2010 and 2020.
Similarly, the analyst can study the history of more complex views,
where each view contains  only the calls up to certain duration, say for $\le$ 1, 5, or 10 minutes.
A classic example of such analyses from literature is
 reference~\cite{leskovec:diameter} that studied the component size, vertex degrees, and diameters of different time-windows in time-stamped citation and web graphs and under different selection criteria of vertices, e.g., those belonging to a particular component or without incoming edges. In other settings, applications may study the history of social or e-commerce networks to find the trends in the centralities or importance rankings of nodes across different snapshots.}
\end{example}

\begin{example}\label{ex:contingency-analysis}
\emph{Perturbation or contingency analysis is a popular analysis done on real-world graphs to study the resilience of graphs to different failure or perturbation scenarios.
For example, in network analyses in neuroscience, scientists ``lesion'' anatomical or functional brain networks by deleting nodes or edges randomly or in a targeted way~\cite{sporns:brain}, e.g., by deleting subsets of the highest degree nodes, and study the effects of these lesions on the average path lengths between different nodes in these graphs. Similarly, our recent user survey~\cite{sahu:extended-survey} reported an application~\cite{zhao:contingency-analysis} from power grids, which are modeled as graphs.
 The application periodically takes a static snapshot of the grid and constructs multiple views of this graph, each representing a failure scenario through the removal or updates
of sets of nodes or edges.
Several computations, such as power or path analysis, are performed to analyze
the effects of each scenario.
Similar contingency analysis applications have also been described in references from many other fields, such as communication~\cite{sterbenz:network-resilience}, transportation~\cite{ip:transportation-resilience}, or other biological networks~\cite{yadav:nexcade}.}
\end{example}

\noindent  These and many others applications require constructing multiple, sometimes hundreds of, views of a static input graph, and compute the same computations across each view. Without a system support, users 
would need to resort to running these computations from scratch on each individual view, which can be very inefficient.
 A system that is able to share computation across views would be of immense use in enabling the efficient development of these applications. We have developed an open-source new analytics system we call \graphsurge{} for this purpose.\footnote{Code is available at \url{https://github.com/dsg-uwaterloo/graphsurge}.}
\graphsurge{} is a full-fledged analytics system that treats graph views as first-class.
The system has a declarative view definition language called \emph{GVDL} with which users can define: (1) individual views; or (2) collections of graph views, which we call {\em view collections}.
Users program \graphsurge{} by writing batch static analytics computations using a dataflow-based API\@. When users execute their programs on view collections, \graphsurge{} automatically shares computation across the views to improve performance by leveraging {\em differential computation}, which can have significant performance benefits.
For example, in a historical analysis application that analyzes
the evolution of a Stack Overflow dataset over 5 years, running a strongly connected components algorithm from scratch  takes 431s, while the same analysis takes merely 43s using \graphsurge{} (see Section~\ref{sec:exp-diff-scratch}).

\graphsurge{} is developed on top of the {\em Timely Dataflow}~\cite{murray:naiad,timely-dataflow-rust} system and its {\em Differential Dataflow} layer~\cite{differential-dataflow-rust, mcsherry:differential-dataflow}, which implements the differential computation technique~\cite{abadi:differential-foundations}.
Differential computation is a black-box technique to incrementally maintain arbitrary, possibly iterative, static dataflow programs, across evolving data sets. Prior literature has used differential computation primarily for maintaining streaming (i.e., continuous) computations for evolving datasets, e.g., to maintain relational queries over a changing database~\cite{gobel:relational-differential-dataflow}.
By static dataflow programs, we refer to those that are designed to run on static datasets and do not contain any incrementalization code. Differential computation is a powerful technique that can automatically incrementalize such static programs.

Our approach is based on the observation that although \graphsurge{} processes static graphs, one can organize view collections as \emph{edge difference sets} that represent them similar to an evolving graph. Specifically, \graphsurge{} first orders
a view collection $C$ with $k$ views and gives each view an index $GV_1$, ..., $GV_k$ (this step is discussed momentarily). 
Then, the system runs a Timely Dataflow program that first materializes $GV_1$ and for each view $GV_i$, materializes only $GV_i$'s {\em edge differences}, i.e., edge additions and deletions, compared to $GV_{i-1}$. This edge difference-based storage compactly materializes all of the views in the $C$ and represents $C$ as an evolving graph over $k$ time steps. 
Finally, when running the same analytics computation across
the views of $C$, \graphsurge{} feeds the user program and the
computed difference sets for each $GV_i$ to Differential  Dataflow, which shares
computation across views internally by running the program differentially across the views.

Unlike streaming applications on Differential Dataflow or specialized graph streaming systems, such as GraphBolt~\cite{mariappan:graphbolt},  %
 the static nature of the views defined in \graphsurge{} gives the system several interesting optimization opportunities, which we next discuss.%
\vspace{5px}
\noindent {\bf Collection Ordering Problem:}
Intuitively, once a view collection is ordered as consecutive edge difference sets, making the neighboring views more similar
allows differential computation to share more computation across views.
In streaming or continuous query processing systems, a system
has no choice over the order of the updates that come in, so effectively no choice
as to the order of the snapshots on which a computation has to be performed.
Instead, the static nature of the views gives \graphsurge{}
an opportunity to order the views as a preprocessing step
and put similar views close to each other.
We show that this problem is NP-hard, but
show a constant-factor approximation algorithm that we have integrated into \graphsurge{}.
In our evaluations, we show that our collection ordering optimization can lead to up to 10.1x runtime improvements when good orderings are unclear.%

\vspace{5px}
\noindent {\bf Collection Splitting Problem:} Even after a system has found a good ordering that minimizes differences between views and maximizes computation sharing, there are cases when differentially maintaining the computation for a view $GV_j$, given the differential computations for $GV_0,...,GV_{j-1}$, might be slower than rerunning $GV_j$ from scratch. We call this the collection splitting problem, as rerunning the computation from scratch at $GV_j$ effectively splits the view collection into 2 sub-collections, each of which would be run differentially (in absence of further splittings). Several factors that can trigger this behavior, such as the analytics computation that is executed may be unstable or the views may not be similar enough to benefit from differential computation sharing. 
We show that a system can monitor the runtimes of each view and the sizes of the edge differences, and make effective decision to decide whether to run each view from scratch or differentially. We show that our collection splitting optimization can detect cases when running all views differentially or from scratch is optimal, and lead to up to 1.9x performance  improvements over the better of these baselines when neither is optimal.

Our contributions are as follows. We developed \graphsurge{}, an end-to-end open-source graph
analytics system that is designed for applications that perform batch analytics computations 
over multiple views of the same graph. \graphsurge{} is developed on top of the Timely and Differential Dataflow systems.
A key component of \graphsurge{} is the support for organizing multiple views into 
view collections, which are materialized in a compact manner as edge difference sets.
View collections allow \graphsurge{} to share computation across the views by running a computation differentially 
on all of the views in a collection. We identify and provide techniques for two unique optimization problems that arise in the context of \graphsurge{}: (i) collection ordering, which orders the views in a view collection to minimize the number of edge differences; and (ii) collection splitting, which adaptively decides to run certain views from scratch instead of differentially.

\vspace{-5pt}
\section{Background}\label{sec:background}

\noindent {\bf Property Graph Model:}
\graphsurge{} uses the property graph model, where data consists of a set of nodes and directed edges and arbitrary key-value properties on nodes and edges. Our current implementation supports string, integer, and boolean properties.

\noindent {\bf Timely Dataflow (TD)~\cite{timely-dataflow-rust, murray:naiad}:} TD is a system for 
general, possibly cyclic, i.e., iterative, data-parallel computations that are expressed as
a combination of \emph{timely operators}, such as \emph{map}, \emph{reduce}, and \emph{iterate}, that transform one or more input data streams 
to an output stream. TD is  an inherently streaming system but supports bulk synchronous computations by giving programs the ability to synchronize operators at different 
{\em timestamps}, which are vectors of integers, ${<}i_1, i_2, \dots, i_k{>}$, 
where each $i_j$ can represent different nested
iterations of the computation or versions of input data streams (an important feature for differential computations).
Similar
to systems such as MapReduce and Spark, TD
automatically scales  computations to multiple workers, within or across compute nodes, where each worker
processes only a partition of the data.
\graphsurge{} uses TD directly to create individual views and view collections, and indirectly by using Differential Dataflow to run computations. 

\noindent {\bf Differential Dataflow (DD)~\cite{mcsherry:differential-dataflow, differential-dataflow-rust}:}
DD is a system  built on
top of TD for incrementally maintaining the outputs of arbitrary dataflow computations over evolving  inputs. DD is based on the differential
computation model~\cite{abadi:differential-foundations}. Consider the Bellman-Ford~\cite{cormen:clrs} algorithm for computing single source shortest paths (SSSP) from a source $s$ to all other vertices in a graph $G$. Let $c(u, v)$ be the cost of an edge in $G$. Initially $s$ has a distance of $0$ and every other vertex has a distance of $\infty$. Iteratively, until a fixed point, each vertex $x$ whose distance has changed produces for each of its outgoing neighbor $y$ a possible distance ``message'' $d(x) + c(x, y)$. Vertices update their distances by taking the minimum of their latest distance and these messages.  Figure~\ref{fig:dataflow} shows a dataflow implementation of this computation consisting of two original inputs, \texttt{Edges} (\texttt{E}) and \texttt{Distances} (\texttt{D}), and two operators: (i) a \texttt{JoinMsg} operator taking as input edge tuples $(u, v, c(u,v))$ and latest vertex distances  and outputting the messages \texttt{M}; (ii) a {\em UnionMin} operator taking latest distances and messages for each $v$ and producing (possibly new) distances. 

\begin{figure}[!t]
\centering
\includegraphics[width=0.4\textwidth]{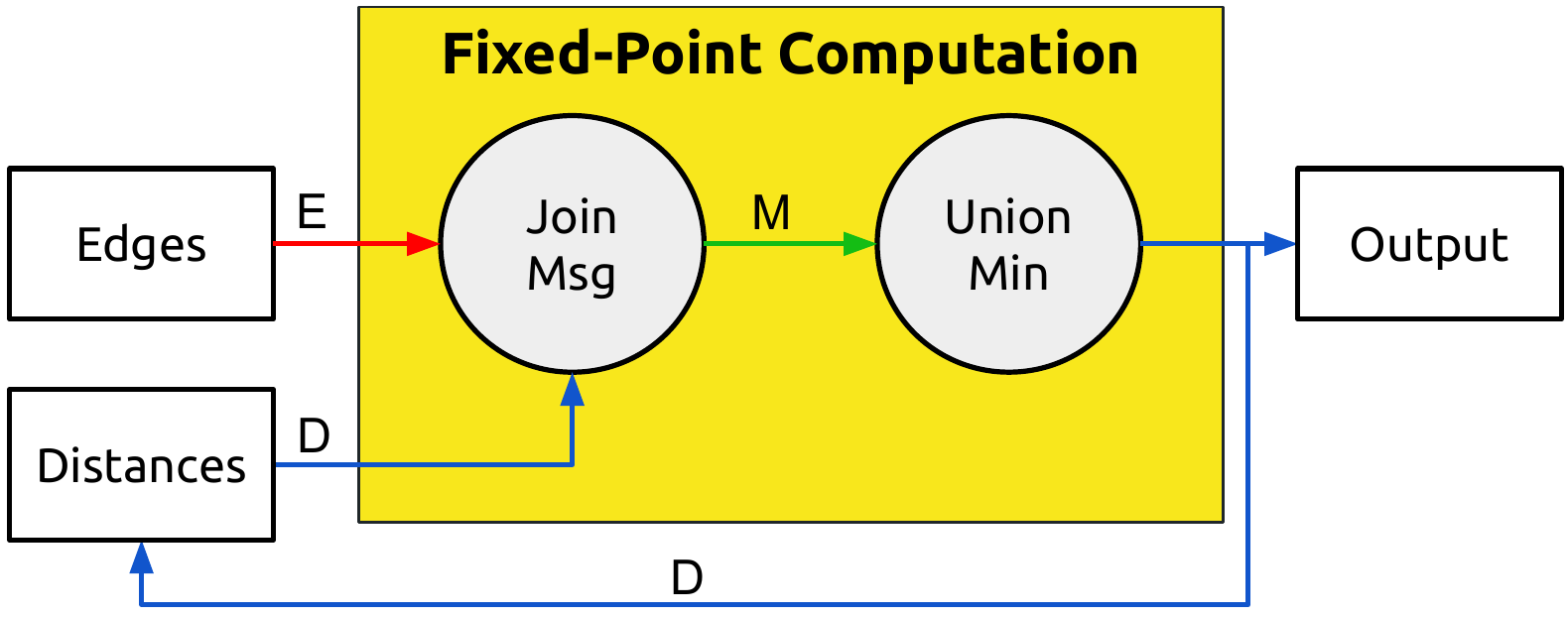}
\vspace{-7pt}
\caption{Dataflow of the Bellman-Ford algorithm for SSSP.}\label{fig:dataflow}
\vspace{-10pt}
\end{figure}

Given a dataflow computation, DD stores the state of the input and output data streams of each operator as \emph{partially ordered timestamped differences} and maintains these differences as the original inputs to the dataflow, e.g., stream \texttt{E} in our example, change. In the above computation, the timestamps are two-dimensional ${<}\text{\emph{graph-version}}$, $\text{\emph{SSSP-iteration}}{>}$ tuples because the streams, specifically \texttt{D}, can change for two separate reasons: (1) changes to \texttt{E}; and (2) changes between iterations of the Bellman-Ford computation. We note that for each timestamp there is a set of differences, and these sets are partially ordered. However, there is {\em no order} among the differences with the same timestamp.

\begin{figure}[!t]
    \centering
    \includegraphics[width=0.31\textwidth]{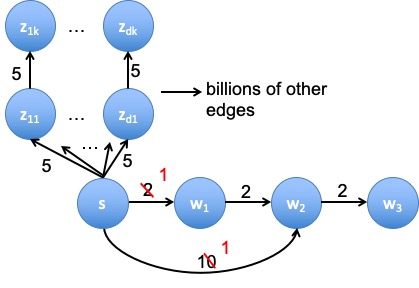}
    \vspace{-12pt}
    \caption{Example stream \texttt{Edges} (\texttt{E}) for the dataflow.}%
    \label{fig:running-ex-input}
    \vspace{-10pt}
    \end{figure}

For a stream $S$, let $S_t$ represent the state of $S$ at timestamp $t$ and let $\delta S_t$ be the difference to $S$ at $t$ (defined momentarily). Consider an operator with a single input stream $A$ and output stream $B$.  DD only keeps track of the differences $\delta A_t$ and $\delta B_t$ ensuring $A_t$ and $B_t$ can be constructed for each $t$ by summing their differences {\em prior} to $t$ according to the partial order of the timestamps, i.e., $A_t = \cup_{s \le t} \delta A_s$ and  $B_t = Op(\cup_{s \le t} \delta A_s)$. 
These equations imply that $\delta A_t = A_t -  \cup_{s < t} \delta A_s$ and $\delta B_t = Op(\cup_{s \le t} \delta A_s) - \cup_{s < t} \delta B_s$, which is how DD computes and stores $\delta A_t$ and $\delta B_t$. Streams in DD are multisets of tuples and the tuples in $\delta S_t$ can have negative multiplicities, implying deletions of tuples.
Table~\ref{fig:bellman-ford-diffs} shows the example of differences to the  $E$, $D$, and $M$ streams in the Bellman-Ford dataflow as the graph in Figure~\ref{fig:running-ex-input} is updated first by changing $(s, w_1)$'s cost from $2$ to $1$ and then $(s, w_2)$'s cost from $10$ to $1$. The input graph contains billions of edges among the $z_{jk}$ vertices, and we denote the difference sets relevant to them in Table~\ref{fig:bellman-ford-diffs} by $\delta Z_E$, $\delta Z_D$, and $\delta Z_{M}$. 
Readers can verify that in Table~\ref{fig:bellman-ford-diffs}, $S_t = \Sigma_{s \le t} \delta S_s$ for every stream and $t$ for the $w_i$ component of the graph.

\begin{table}[t]
    \centering
    \setlength\tabcolsep{1.0pt}
    \setlength\extrarowheight{1pt}
    \small
    \begin{tabular}{@{\hskip-8pt}c|@{}c@{}|c|L{3.5cm}|L{1.5cm}|L{1.7cm}|} 
    \multicolumn{6}{c}{
    \begin{tikzpicture}
    \bfseries
    \draw[->, line width=0.2mm] (-2,0) --node  [above] {\small{Time/Graph Updates}} (2,0);
    \end{tikzpicture}
    } \\\cline{2-6}
    &&&\centering{\textbf{$G_0$}} &  \centering{\textbf{$G_1$}} & \multicolumn{1}{|>{\centering\arraybackslash}m{1.5cm}|}{\textbf{$G_2$}} \\\cline{2-6}
    \multirow{12}{*}{
    \vspace{-2.5cm}
    \begin{tikzpicture}
    \bfseries
    \centering
    \draw[<-, line width=0.2mm] (0,-2) --node [midway, above, sloped] {\small{B-Ford iterations}} (0,2);
    \end{tikzpicture}
    }
    &\multirow{3}{*}{\textbf{0}} & \textcolor{red}{\textbf{$\delta E$}} 
                        &$+(s, w_1, 2),$ $+ (s, w_2, 10)$,
                        $+(w_1, w_2, 2), dZ_E$
                        &\textcolor{red}{$-(s, w_1, 2)$}, $+(s, w_1, 1)$ 
                        &\textcolor{red}{$-(s, w_2, 10)$}, 
                        $+(s, w_2, 1)$ \\\cline{3-6}
                    &   & \textcolor{blue}{$\delta D$} & $+(s,0), + (w_1, \infty),$
                        $+(w_2, \infty),+(w_3, \infty),$ $dZ_D$& $\varnothing$& $\varnothing$ \\\cline{3-6}
                    &  &\textcolor{teal}{\textbf{$\delta M$}} &$+(w_1, 2),+(w_2, 10),$ $dZ_M$ & \textcolor{red}{$-(w_1, 2)$}, $+(w_1, 1)$ 
                        & \textcolor{red}{$-(w_2, 10)$}, $+(w_2, 1)$\\\cline{2-6}
    &\multirow{3}{*}{\textbf{1}} & \textcolor{red}{\textbf{$\delta E$}} 
                        & $\varnothing$
                        &$\varnothing$ 
                        &$\varnothing$ \\\cline{3-6}
                        & & \textcolor{blue}{$\delta D$} 
                        & \textcolor{red}{$-(w_1, \infty)$}$, +(w_1, 2)$,
                        \textcolor{red}{$-(w_2, \infty)$}, $+(w_2, 10)$, $ dZ_D$
                        & \textcolor{red}{$-(w_1, 2)$}, $+(w_1, 1)$
                        & \textcolor{red}{$-(w_2, 10)$}, $+(w_2, 1)$ \\\cline{3-6}
                        & &\textcolor{teal}{\textbf{$\delta M$}} 
                        &$+(w_2, 4),+(w_3, 12),$ $dZ_M$
                        & \textcolor{red}{$-(w_2, 4)$}, $+(w_2, 3)$ 
                        & \textcolor{red}{$-(w3, 12)$}, $+(w_3, 3)$\\\cline{2-6} 
    &\multirow{3}{*}{\textbf{2}} & \textcolor{blue}{\textbf{$\delta D$}} 
                        & \textcolor{red}{$-(w_2, 10)$}, $+(w_2, 4)$,
                        \textcolor{red}{$-(w_3, \infty)$}, $+(w_3, 12)$, $ dZ_D$
                        & \textcolor{red}{$-(w_2, 4)$}, $+(w_2, 3)$
                        &  $+(w_2, 10)$, \textcolor{red}{$-(w_2, 3)$},
                        \textcolor{red}{$-(w_3, 12)$}, $+(w_3, 3)$ \\\cline{3-6}
                        & &\textcolor{teal}{\textbf{$\delta M$}} 
                        &\textcolor{red}{$-(w_3, 12)$}, $+(w_3, 6)$, $dZ_M$
                        & \textcolor{red}{$-(w_3, 6)$}, $+(w_3, 5)$ 
                        & \textcolor{red}{$-(w_3, 5)$}, $+(w_3, 12)$\\\cline{2-6} 
    &\multirow{3}{*}{\textbf{3}} & \textcolor{blue}{\textbf{$\delta D$}} 
                        & \textcolor{red}{$-(w_3, 12)$}, $+(w_3, 6)$,  $dZ_D$
                        & \textcolor{red}{$-(w_3, 6)$}, $+(w_3, 5)$
                        & \textcolor{red}{$-(w_3, 5)$}, $+(w_3, 12)$ \\\cline{3-6}
                        & &\textcolor{teal}{\textbf{$\delta M$}} 
                        & $dZ_M$ & $\varnothing$ & $\varnothing$ \\\cline{2-6}
    &... & ... & rest contains $dZ_D$, $ dZ_M$
                        &rest is $\varnothing$ 
                        &rest is $\varnothing$ \\\cline{2-6}
    &k & ... & ... &... &... \\\cline{2-6}
    \end{tabular}
    \vspace{4pt}
    \caption{Differences in the SSSP example. $\delta E$ is $\emptyset$ and omitted after SSSP iteration 0 in each column.}%
    \label{fig:bellman-ford-diffs}
    \vspace{-25pt}
    \end{table}

We end this section with three important properties of differential computation and DD:

\noindent {\em Property 1: DD programs are written for static input datasets.} The Bellman-Ford dataflow program from Figure~\ref{fig:dataflow} computes shortest paths on a {\em static} graph. Specifically, it contains no logic for maintaining the computation if the input graph changed. The primary utility of DD is that it can incrementalize any computation without requiring programmers to write any special code to incrementalize their programs, which is a very hard task for many programmers. 

\vspace{2pt}
\noindent {\em Property 2: Maintaining computations differentially can save a lot of computation.} If there are no differences to the inputs of an operator $O$ at version $<$$G_{j}, i$$>$, $O$'s output is guaranteed to have an empty difference, which allows DD to maintain the differences without performing any computation. Indeed, after the $G_0$ column is computed,
Table~\ref{fig:bellman-ford-diffs} shows all the merely 30 updates to differences that DD computes, despite the fact that we assumed that the graph contains billions of edges. This is because, DD automatically notices that the results of computations working on data stream partitions related to vertices $z_{ij}$ effectively cannot have changed after updates. As a result, DD avoids rerunning any computation for those parts, effectively sharing computation across the three versions of $G$.

\vspace{2pt}
\noindent {\em Property 3: DD can also save computation when running iterative programs on static input datasets that run until a fixed point.} 
For example, even if the graph was not changing, implementing Bellman-Ford algorithm in dataflow systems requires a logic that is similar to differential computation, where some user-specific code has to check whether the computation has reached a fixed point.%

\vspace{-5pt}
\section{Graphsurge System}\label{sec:graphsurge}

\graphsurge{} is a system for performing analytics on views over static input graphs. The system is implemented in Rust. Figure~\ref{fig:architecture} shows the architecture of \graphsurge{}. Users program \graphsurge{} through two interfaces: (i)  A declarative {\em graph view definition language} (GVDL) to define individual views and view collections over base graphs; and (ii) A DD-based API to write dataflow programs for graph analytics computations that consume the difference stream of a view or view collection. \graphsurge{} uses TD and DD as its execution layer for both creating views as well as for running user-specified analytics programs on views. As such, dataflows written in TD and DD can be automatically parallelized both in a single multi-core machine and in a distributed cluster.

Users import \emph{base input graphs} to \graphsurge{} through CSV files that contain the nodes and edges of the graph and their properties. Upon loading, nodes and edges are given unique 32-bit IDs and stored as a \emph{node stream} and {\em edge stream} in the \emph{graph store} (GStore), respectively. Each edge in an edge stream is a (\texttt{eID}, \texttt{sID}, \texttt{dID}, \texttt{key$_1$}, \texttt{val$_1$}, $\ldots$) tuple, where \texttt{eID} is the edge ID\@, \texttt{key$_i$} and \texttt{val$_i$} are the key-value properties of the edge, and \texttt{sID} and \texttt{dID} are source and destination node IDs that point to offsets in the node stream.

When a user runs a GVDL query, a TD program is executed that reads the edge stream from the GStore, applies the filter predicates to generate the difference sets for the corresponding view or view collection, and stores them in the \emph{view and collection store} (VCStore). When the user then runs an analytics computation on a view or view collection, a DD program reads the difference sets from the VCStore to run the user-defined dataflow. In a distributed cluster, both GStore and VCStore are replicated across all the machines. Each thread of a running TD and DD dataflow operates on a logical partition of the edge stream when creating views or view collections, and of the difference sets when running an analytics computation. The parallel running operators in the TD and DD dataflows perform read-only queries to GStore and VCStore and do not require any locks or coordination.

\begin{figure}[t]
    \centering
    \includegraphics[width=\columnwidth]{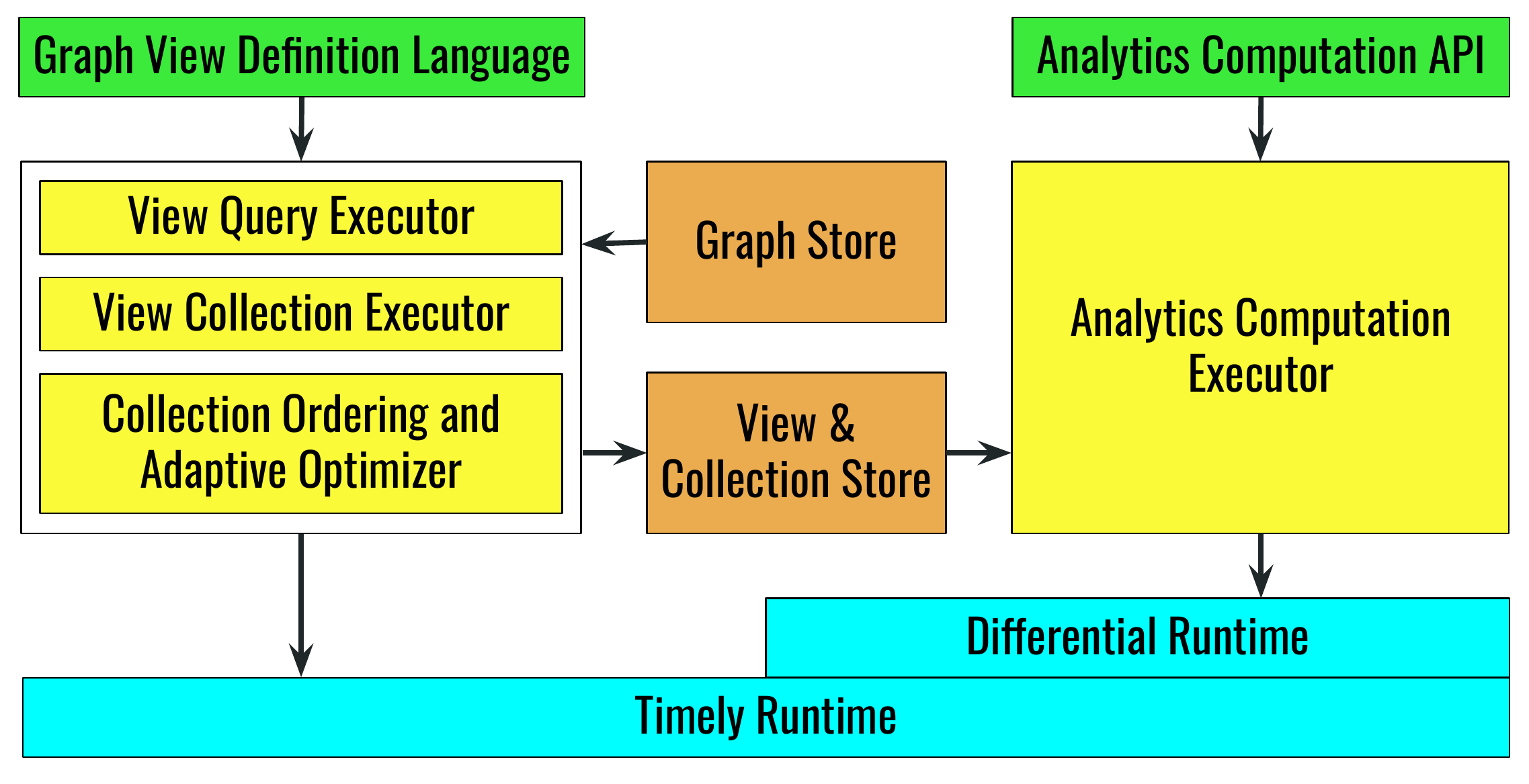}
    \vspace{-15pt}
    \caption{\graphsurge{} architecture.}%
    \label{fig:architecture}
    \vspace{-5pt}
\end{figure}

\vspace{-5pt}
\subsection{Individual Views}%
\label{subsec:individual-views}

\subsubsection{Individual View Definition}
GVDL is a simple language to define views
over base graphs. 
A GVDL query to create views has a single \texttt{Where} clause that specifies a  predicate on an input graph (or another materialized view) that specifies the edges of the output view.  Predicates can be arbitrary conjunctions or disjunctions and access the properties of both source and destination nodes as well as the edges.

\begin{example}\label{ex:filtered-view} 
{\em
Listing~\ref{lst:filtered-view} shows a view an analyst can construct on 
our running example \texttt{Calls} graph. The view
consists of calls made in California in 2019 with duration longer than 10 minutes.
}
\end{example}

\begin{lstlisting}[float, 
    caption={Example GVDL view query.},
    label={lst:filtered-view}, belowskip=-0.8 \baselineskip]
create view CA-Long-Calls on Calls
edges where src.state = 'CA' and dst.state = 'CA'
        and duration > 10 and year = 2019
\end{lstlisting}

The views that users can express in GVDL are noticeably simple, but they 
are enough to express the use cases we have explored
in this paper. In Section~\ref{subsec:view-collections}, we will elaborate 
on a second advantage of keeping GVDL simple when we discuss
how we store multiple views that are organized in a view collection in a compact manner.
GVDL queries that define individual views are compiled into TD dataflow programs in a straightforward fashion. The dataflow consists of a filter operator to apply the user-specified predicates to the edge stream and compute the difference sets. The output of the program is materialized as a stream in the VCStore.

\subsubsection{Analytics Computations on Individual Views}%
\label{subsec:analytics-api}
\begin{lstlisting}[float,
    caption={Differential Computation API.},
    label={lst:diff-comp-api}, belowskip=-0.8 \baselineskip]
pub trait GraphSurgeComputation {
    type Results;
    fn graph_analytics(input_stream: &InputStream)
        -> Collection<Self::Results> }
\end{lstlisting}

Users write arbitrary DD dataflow programs for performing analytics on their views with the constraint that one of the inputs to the dataflow is the \graphsurge{}-specific difference stream for the view. \graphsurge{} exposes a Rust interface to users with a \texttt{graph\_analytics} function, inside which users can write arbitrary DD programs that are expected to return per-vertex user-defined outputs, such as the connected component ID of each vertex in a connected components analytics. Listing~\ref{lst:diff-comp-api} shows the interface of the \texttt{graph\_analytics} function. Users invoke their programs through a separate command line and specify their \texttt{graph\_analytics} function and the view on which to run this function. \graphsurge{}'s \emph{analytics computation executor} calls the users' \texttt{graph\_analytics} function to obtain a computation dataflow, and feeds the edge stream corresponding to the view into it. When the computation is executed on a single view, the entire edge stream is fed into the dataflow at once. How the computation is executed on a view collection is more involved and described in Section~\ref{subsec:view-collections}.

\vspace{-5pt}
\subsection{View Collections}\label{subsec:view-collections}

\subsubsection{View Collection Definition}\label{subsubsec:view-collections}

To share analytics computations across multiple views of the same graph, \graphsurge{} allows users to organize views in a {\em view collection}. A view collection organizes a set of views as a single timestamped\footnote{We use the term timestamp to follow differential computations' terminology. This should not be confused with any application-specific ``time'' property, such as the \texttt{year} property we use in our running example.} {\em edge difference stream}~$C$, where each view corresponds to a state of the stream at a particular timestamp~$t$.

\begin{lstlisting}[float,
    caption={Example GVDL view collection query.},
    label={lst:ex-view-collection-gvdl}, belowskip=-0.8 \baselineskip]
create view collection call-analysis on Calls
    [$GV_1$: ID $<$ 100],
    [$GV_2$: ID $\ge$ 50 and ID $<$ 199],
    [$GV_3$: ID $\ge$ 10 and ID $<$ 100],
    [$GV_4$: ID $\ge$ 60 and ID $<$ 199]
\end{lstlisting}

\begin{example}\label{ex:view-collection} 
{\em
Listing~\ref{lst:ex-view-collection-gvdl} shows a simple demonstrative GVDL query defining a view collection with four views on our \texttt{Calls} graph. Each view including all calls within a range of edge IDs.
} 
\end{example}

\graphsurge{} materializes the view collection described by a GVDL query 
in three steps. Below, we let $p_j$ denote the predicate defining $GV_j$ in a given view collection.

\begin{figure}[t]
    \centering
    \setlength\tabcolsep{0.7pt}
    \begin{tabular}{@{}ccccclrrrrlrrrr@{}}
        \toprule
        \multicolumn{1}{l}{}      & \multicolumn{4}{c}{$EBM$}   &  & \multicolumn{4}{c}{$EDS_{def} $ (\# = 540)}                               &  & \multicolumn{4}{c}{$EDS_{opt}$ (\# = 260)}                               \\ \cmidrule(lr){2-5} \cmidrule(lr){7-10} \cmidrule(l){12-15} 
        \multicolumn{1}{l}{edges} & $GV_1$ & $GV_2$ & $GV_3$ & $GV_4$ &  & $GV_1$        & $GV_2$        & $GV_3$        & $GV_4$        &  & $GV_3$        & $GV_1$        & $GV_2$        & $GV_4$        \\ \midrule
        $e_0$-$e_9$               & 1      & 0      & 0      & 0      &  & +1            & -1            & $\varnothing$ & $\varnothing$ &  & $\varnothing$ & +1            & -1            & $\varnothing$ \\
        $e_{10}$-$e_{49}$         & 1      & 0      & 1      & 0      &  & +1            & -1            & +1            & -1            &  & +1            & $\varnothing$ & $\varnothing$ & -1            \\
        $e_{50}$-$e_{59}$         & 1      & 1      & 1      & 0      &  & +1            & $\varnothing$ & $\varnothing$ & -1            &  & +1            & $\varnothing$ & $\varnothing$ & $\varnothing$ \\
        $e_{60}$-$e_{99}$         & 1      & 1      & 1      & 1      &  & +1            & $\varnothing$ & $\varnothing$ & $\varnothing$ &  & $\varnothing$ & $\varnothing$ & +1            & $\varnothing$ \\
        $e_{100}$-$e_{199}$       & 0      & 1      & 0      & 1      &  & $\varnothing$ & +1            & -1            & +1            &  & $\varnothing$ & $\varnothing$ & +1            & $\varnothing$ \\ \bottomrule
        \end{tabular}
    \caption{An example EBM for the view collection in Listing~\ref{lst:ex-view-collection-gvdl} and 2 EDS's for 2 different collection orders (Section~\ref{sec:collection-ordering}).}%
    \label{fig:ordering-example}
\vspace{-15pt}
\end{figure}

\noindent {\bf Step 1. {\em Edge Boolean Matrix Computation}:} For each edge $e_i$ in the base graph and each view $GV_j$ in the collection, \graphsurge{} runs the predicate $p_{j}$ on $e$ and outputs an {\em edge boolean matrix} (EBM) that specifies whether $e_i$ satisfies $p_j$. 
This is an embarrassingly parallelizable computation and is performed by a TD dataflow.

\noindent {\bf Step 2. {\em Collection Ordering}:} 
The goal of this step is to put views with a higher overlap of their edges next to each other, so that the differences between neighboring views is smaller and running an analytics computation on the collection results in higher computation sharing. To achieve this, \graphsurge{} re-orders the views in EBM so that views whose predicates satisfy highly overlapping sets of edges are adjacent to
each other. As we discuss momentarily, the goal of this optimization is to
store the views in the collection in a more compact {\em edge difference stream}, i.e.,
using fewer edge differences. As we demonstrate
in our evaluations, this optimization step can lead to significant performance benefits.
The output of this step is the same EBM but possibly with a different column ordering. We defer the details of how collections are ordered to Section~\ref{sec:collection-ordering}.

\noindent {\bf Step 3. {\em Edge Difference Stream (EDS) Computation}:} Finally, \graphsurge{}
takes the reordered EBMs and materializes the views in the view collection as an edge difference stream that is consistent with the semantics of difference sets of differential computation.
Specifically, we treat the entire view collection $C$ as an evolving input stream according to the order obtained in step 2. For simplicity, let $GV_1, \ldots, GV_k$ be the order of the views after step 2, so $C_t = GV_t$.  
Recall from Section~\ref{sec:background} 
that according to differential computation semantics, the difference of a stream $A$ at timestamp
$t$ is $\delta A_t = A_t -  \cup_{s < t} \delta A_s$. So the edge difference of a view $t$,
 $\delta C_t$ is computed to ensure that $\delta C_t$$=$$GV_t$$-$$\cup_{s < t} \delta C_s$ equality holds. Specifically, the multiplicity of each edge $e_i$ in $\delta C_t$ is: (i) 0 if $GV_{t-1}$ and $GV_t$ both contain or both do not contain $e_i$; (ii) 1 if $GV_{t-1}$ does not contain $e_i$ and $GV_t$ does; or (iii) -1 if $GV_{t-1}$ contains $e_i$ and $GV_t$ does not.
The contribution of each edge $e_i$ to $\delta C_t$ can be computed independently, so this is another embarrassingly parallelizable step. 

\begin{example}%
\label{ex:ebm-eds}
Figure~\ref{fig:ordering-example} shows an example EBM for the 
view collection from Listing~\ref{lst:ex-view-collection-gvdl}. For example, $GV_1$ has
1 for all the edges $e_0$ to $e_{99}$ and 0 for others since its predicate is $ID < 100$.
Ignore the right side of the figure for now. On the left side, the figure also shows the EDS
$EDS_{def}$ that corresponds to the default order of $GV_1$,$GV_2$,$GV_3$,$GV_4$.
The first row of $EDS_{def}$ for $e_0$-$e_9$ contains: (i) +1 for $GV_1$ because $GV_1$ contains all of these
edges; (ii) -1 for $GV_2$ because $GV_2$ does not contain any edges (so that union of these
differences with $GV_1$ gives the empty set); and $\emptyset$ for the rest of the views 
because they also do not contain these edges.
\end{example}

We end this section with a note on GVDL. Recall from Section~\ref{subsec:individual-views} that we have limited the 
view queries users can express in GVDL to simple node and edge filter
predicates. This ensures that each view over the same base graph contains
a subset of a larger ``ground truth'' set of edges and that
each view has the same set of node IDs, i.e., 
a node with ID $u$ in a view $GV_i$ maps to the node with ID $u$ in a view $GV_j$. 
This allows \graphsurge{} to easily compute an EBM for a collection
and the edge differences between two views from the EBM. If 
we allowed views that created new nodes, and we could not assume
a node ID mapping between the views, the system could not easily compute 
an edge boolean matrix or difference stream, which is critical for us to 
store the views compactly and use differential computation when we run
analytics computations over view collections (discussed in the next section). 
One can extend GVDL to support more general {\em individual views} that
can create new nodes and edges, e.g., those that form super nodes and edges, 
and run analytics computations over these views. However, 
it would be challenging to store multiple such views compactly in a view collection as an edge difference stream,
if the system cannot infer a mapping between the nodes across views.

\vspace{-3pt}
\subsubsection{Analytics on View Collections}\label{subsubsec:view collection-analytics}

Given an analytics program~$P$ that a user wants to run on all views of a view collection~$C$, in absence of any collection splitting, which is an optimization we describe in Section~\ref{sec:collection-splitting}, the analytics computation executor runs~$P$ as follows. First, the system runs~P on $C_{0}$, i.e., the ``first'' view in $C$, and when this computation finishes, in an outside loop {\em advances} (in DD terminology) $C$ to $C_{1}$ by feeding $\delta C_1$ to DD. Then the system feeds $\delta C_{2}$ to DD, so on and so forth, until all views are evaluated. When computing $P$ at each time $t$, DD will automatically share computation from the ``prior'' views on which $P$ has been computed, in some cases leading to significant performance gains compared to running~$P$ on each view from scratch. The output of the DD program is a set of {\em output difference sets} for the output \texttt{(VID, Results)} stream specified in the \texttt{graph\_analytics} function. The output difference stream can then be stored or processed by the user.

\vspace{-5pt}
\subsubsection{Support for dynamic graphs}\label{subsubsec:dynamic-graphs}

While \graphsurge{} is built for applications that work with static views of a static graph, it can also support analysis of dynamic graphs by ingesting timestamped stream of updates to a graph and creating a view collection where the views represent batches of updates for different time windows given by a filter predicate on the edge timestamps. However, in our current implementation, the analytics are not performed in a  traditional streaming fashion and all of the data a users want to analyze needs to be fully ingested into \graphsurge{} before creating view collections and running their analytics.

\section{Collection Ordering}%
\label{sec:collection-ordering}

Given a set of $k$ views in a view collection $C$ defined by an application, there are $k!$ different ways \graphsurge{} can order the views before running analytics computations differentially on the collection.
This is important because the number of edge differences that are generated in the final collection is solely determined by the order of the views. 
Recall from Section~\ref{subsubsec:view collection-analytics} that when running analytics computations
on a view collection $C$, \graphsurge{} iterates over neighboring views and for view $t$ feeds in the difference set $\delta C_t$ to DD (in absence of collection splitting). The smaller the size of the differences, the larger the structural overlap
between view $C_t$ and the union of the views prior to $C_t$,
which we expect to lead to larger computation sharing. As we present in our evaluations, by picking orderings that minimize the set of differences, \graphsurge{} can improve performance significantly in certain applications.
 We can formulate this problem as a concrete optimization problem as follows:

\begin{definition} Collection Ordering Problem (COP): {\em Given a view collection $C$,
find the collection ordering that minimizes the sum of the sizes of difference sets $\delta C_t$.}
\end{definition}

We next show that COP is NP-hard. 
Our proof is through a reduction from the
{\em consecutive block minimization problem (CBMP)} for boolean matrices. In a boolean matrix $B$, such
as the edge boolean matrix (EBM) in Figure~\ref{fig:ordering-example}, a \emph{consecutive block} is a 
maximal consecutive run of 1-cells in a single row of $B$,
which is bounded on the left by either the beginning of the row or a 0-cell, and bounded on the right by
either the end of the row or a 0-cell. Given a column ordering $\sigma$ for $B$, let $cb(B, \sigma)$ denote the total number of consecutive blocks in $B$ over all rows.  CBMP is the problem of finding the ordering $\sigma^*$ of the columns of $B$ that minimizes $\min_{\sigma}cb(B, \sigma)$. 
CBMP is known to be NP-hard~\cite{kou:cbm-np-complete}. 
\begin{theorem} COP is NP-hard.
\end{theorem}
\begin{proof}
For an input graph $G$, let $C$ be a view collection over $G$,
and let the EBM of $C$ be $C_{EBM}$. 
For a fixed column ordering $\sigma$ of $C_{EBM}$, let
the size of the difference sets in $\sigma$ be $ds(C_{EBM}, \sigma)$.
Therefore,
COP is equivalent to the following problem on boolean matrices: given a
boolean matrix $C_{EBM}$, find the $\sigma^*$ that minimizes $\min_{\sigma}ds(C_{EBM}, \sigma)$. 
Recall from step 3 of view collection materialization (Section~\ref{subsec:view-collections}) that 
the difference set for an edge $e$, which is represented by a row $r$ in $C_{EBM}$, is calculated as follows: for the first appearance of $e$ from left to right, i.e., for the first 1 in $r$, we count one difference. Then for each subsequent alternating appearance of a 0,
and then 1, and then 0, etc., we count one additional difference. Note that this is different than the definition of a consecutive block. For example, a row (1110) has 1 consecutive block but 2 diffs: one diff for the first view, and one diff for the last view. 

Our reduction is from CBMP\@. Given a $k_1 \times k_2$ matrix $B$ 
to CBMP, consider (in poly-time) constructing a $2k_1 \times k_2$ matrix $B_{EBM}$ that contains B and the complement of $B$, 
$B^C$, under $B$. That is $B^C$ contains 1s where B contains 0s and vice versa.
Note that for each row $r$ of $B$, both $r$ and $r^C$ appear in $B_{EMP}$ exactly once.
Let $B_0$, $B_1$, and $B_{01}$, respectively, be the set of rows in 
$B$ that contain only 0s, only 1s, and both a 0 and a 1, and let $|B_0|=m_0$, 
$|B_1|=m_1$, and $|B_{01}|=m_{01}$.  
Given an arbitrary column ordering $\sigma$, we analyze the number of differences each row in $B$ induces in $B_{EMP}$.
\begin{squishedlist}
\item Row $r$ in $B_0$ yields 0 but $r^C$ yields 1 difference.
\item Row $r$ in $B_1$ yields 1 difference but $r^C$ yields 0 difference.
\item Row $r$ in $B_{01}$ requires analyzing two cases. Let $cb(r, \sigma)$ denote the
number of consecutive blocks only in row $r$. (i) If $r$'s last cell is a 0, then $r$ yields $2cb(r, \sigma)$ and $r^C$ yields $2cb(r, \sigma)-1$ differences; and (ii) otherwise $r$ yields  $2cb(r, \sigma)-1$ and $r^C$ yields
 $2cb(r, \sigma)$ differences. Therefore, in either case, $r$ yields $4cb(r, \sigma)-1$ differences. 
\end{squishedlist} 

\noindent Therefore, $ds(B_{EBM}, \sigma)$ is: $(\sum_{j\in B_{01}} 4cb(r, \sigma)-1) + m_0 + m_1$, which is equal
to $4cb(B, \sigma) - m_{01} + m_0 + m_1$. This establishes a one-to-one connection between
the sizes of the difference sets in $B_{EBM}$ and the number of consecutive blocks in $B$ under any ordering $\sigma$. 
Since for any $B$, $m_0$, $m_1$, and $m_{01}$ are fixed, finding the optimal ordering $\sigma^*$ 
that minimizes $ds(B_{EBM})$ also minimizes $cb(B)$, completing the proof that solving COP is NP-hard. 
\end{proof}
We next describe a 3-approximation to COP, which uses a 1.5-approximation algorithm for CBMP from reference~\cite{haddadi:cbmp-approx}, CBMP$_{1.5}$, which we next review.
CBMP$_{1.5}$ takes as input an $m \times k$ boolean matrix $B$, creates the matrix $0B$ by padding a 0 column, and then transforms $0B$ into a $(k+1)$ clique $G^{0B}$, where each column (so each view in our case) is a node, and the weight between the nodes is the Hamming distance of the columns they represent. Reference~\cite{haddadi:cbmp-approx} shows that $G^{0B}$ satisfies the triangle inequality and the entire transformation from $B$ to $G^{0B}$ is approximation preserving. Therefore, solving TSP, with the well known Christofides algorithm~\cite{christofides} yields a 1.5-approximation to CBMP after removing the 0 column from the tour, which gives a chain between the
remaining columns to get an ordering.

\begin{corollary} Running CBMP$_{1.5}$ on the EBM of a view collection gives a 3-approximation algorithm for COP\@.
\end{corollary}

\begin{proof}
To see this, consider any input $C_{EBM}$ to COP and any ordering $\sigma$ for $C_{EBM}$. Because each row $r$ contains either $2cb(r, \sigma)-1$ or  $2cb(r, \sigma)$ differences, $cb(C_{EBM}, \sigma) \le ds(C_{EBM}, \sigma) \le 2cb(r, \sigma)$. Therefore, since Christofides algorithm returns a 1.5-approximation algorithm for CBMP, it returns a 3-approximation for COP\@.
\end{proof}

\begin{algorithm}[t]
    \SetAlgoVlined{}
    \SetKwInOut{Input}{input}
    \SetKwInOut{Output}{output}
    \Input{Edge Boolean Matrix $B_{m \times k}$, W workers}
    \Output{A column ordering $\sigma^*$}
    \Begin{
        Partition $B_{m \times k}$ $\rightarrow$ $\bigcup\limits_{i=0}^{W-1} B_{i}$\;
        \Begin(At each worker $w_i, 0\leq{}i<W$:) {
            $C_{i}\leftarrow [0|B_i]$\;
            $U \leftarrow$ unit matrix\;
            $D_i = C_{i}^{T}(U-C_{i}) + (U-C_{i})^{T}C_{i}$\; 
            Shuffle $D_i$ to worker $w_0$
        }
        \Begin(At worker $w_0$:) {
            Receive $D_i$ from all workers $w_i$\;
            $D \leftarrow \sum_{i=0}^{W-1} D_i$\;
            $G^{0B} \leftarrow$  complete graph ($|V|\!=\!(\!k\!+\!1\!)$) induced from $D$\;
            $\sigma^* \leftarrow$ tsp\_christofides($G^{0B}$)\;
            Broadcast $\sigma^*$ to all workers $w_i$\;
        }
    \vspace{-5pt}
    }
    \caption{Collection Ordering Optimizer}\label{algo:ordering-tsp}
\end{algorithm}

Algorithm~\ref{algo:ordering-tsp} shows our collection ordering optimizer. 
We take the EBM $C_{EBM}$ as input. Given $B = C_{EBM}$, we construct $G^{0B}$ using a TD program that performs the padding and then in an embarrassingly parallel way find the Hamming distances between each view. Then we collect the $G^{0CB}$ in  a single worker and run Christofides algorithm in a single TD worker. The output of this algorithm gives 2 possible orders, one for each 
direction of the chain and either is a 3-approximation. However, these orders do not necessarily yield the same
number of differences, and we pick the order with the smaller differences\footnote{We note that 
these orders would return the same value for CBMP.}.

\newcommand{\hhlline}[4]{\draw (#1-#2-#3.south west) -- (#1-#2-#4.south east);}
\newcommand{\vvlline}[4]{\draw (#1-#3-#2.north east) -- (#1-#4-#2.south east);}
\begin{figure}[t]
    \centering
    \begin{tikzpicture}[auto,node distance = 1.5cm]
        \tikzstyle{every state}=[
            draw = black,
            thick,
            fill = white,
            minimum size = 4mm
        ]
        \tikzset{myptr/.style={draw=magenta}}

        \node[state,fill=magenta] (0) at (0,0) {$0$};
        \node[state] (1) at (1.5,1) {$GV_1$};
        \node[state] (2) at (1.5,-1) {$GV_2$};
        \node[state] (3) at (3.5,1) {$GV_3$};
        \node[state] (4) at (3.5,-1) {$GV_4$};

        \path (0) edge node[anchor=south,xshift=-0.15cm] {100} (1);
        \path (0) edge node[anchor=north,xshift=-0.15cm] {150} (2);
        \path (0) edge[myptr] node[anchor=south west,pos=0.25,xshift=-0.1cm] {90} (3);
        \path (4) edge[myptr] node[anchor=south west,pos=0.9,yshift=-0.1cm] {140} (0);
        \path (1) edge[myptr] node[anchor=west] {150} (2);
        \path (3) edge[myptr] node[anchor=south] {10} (1);
        \path (1) edge node[anchor=south west,pos=0.8,xshift=-0.2cm] {160} (4);
        \path (2) edge node[anchor=north west,pos=0.8,xshift=-0.2cm] {140} (3);
        \path (2) edge[myptr] node[anchor=north] {10} (4);
        \path (3) edge node[anchor=west] {150} (4);
    \end{tikzpicture}
    \caption{$G^{0B}$ of the collection from Listing~\ref{lst:ex-view-collection-gvdl}. Purple lines are the output tour/order of the Christofides algorithm.}%
    \label{fig:christofides-example}
    \vspace{-10pt}
\end{figure}
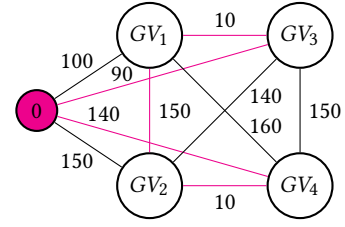

\begin{example}%
\label{ex:view-ordering-tsp}
Figure~\ref{fig:christofides-example} shows the example $G^{0B}$ corresponding to the view collection from
Listing~\ref{lst:ex-view-collection-gvdl}, whose EBM was shown in Figure~\ref{fig:ordering-example}. 
For example, the weight of the edge between $GV_1$ and $GV_3$ is 10 because there 
are only 10 edge differences between these two views, specifically $GV_1$ contains edges $e_0$ to $e_9$,
while $GV_3$ does not and the views overlap on other edges. The red lines in the figure show the TSP tour
that Christofides algorithm outputs. Taking the 0 node out of the tour gives us a chain, where 
the $GV_3$,$GV_1$,$GV_2$,$GV_4$ order is the better of the two possible orders.
The EDS $EDS_{opt}$ that corresponds to this optimized order is shown
on the right side of Figure~\ref{fig:ordering-example}, which reduces the number of differences
of the default order from 540 to 260.
\end{example}

Our collection ordering technique materializes each view in a collection. An interesting question is whether a good 
ordering can be obtained through only inspecting the definitions of the views and without inspecting 
and materializing the views. This can be possible for example when there is a containment relationship between the 
predicates defining the views, in which case the best order follows the containment order. For example, a system can infer 
that a view defined by the predicate ``year < 2010'' is contained within the view defined by ``year < 2011''.  Although 
general query or view containment~\cite{querycontainment} is a hard problem~\cite{querycontainment:hard1,querycontainment:hard2}, prior literature has identified cases when it can be determined, e.g., when the predicates are certain conjunctive queries~\cite{afrati2006rewriting,chekuri:conjunctive}. In cases when views are arbitrary, we do not know of any technique to find a good ordering without inspecting the data in
the views.
Finally, reference~\cite{babu:ordering} has studied ordering predicates of a large conjunctive predicate to put selective predicates earlier. This work assumes a streaming data setting, instead of
our static setting, and orders the predicates 
in a large conjunctive predicate based on general selectivity statistics about the predicates. In contrast to this
work, our goal is to order arbitrary, not-necessarily conjunctive, predicates with the goal of putting predicates
whose outputs have overlaps next to each other.

\vspace{-5pt}
\section{Collection Splitting}%
\label{sec:collection-splitting}
Even after we find a good ordering that minimizes the sizes of the difference sets generated, running each view differentially may not be ideal.
If the computation footprints of $A$ on $GV_i$ and $GV_{i+1}$ are very different, differentially fixing $A$ on $GV_{i+1}$ might be slower than running $A$ on $GV_{i+1}$ from scratch. 

\begin{table}[t]
    \centering
    \begin{tabular}{@{}cccc@{}}
        \toprule
        \textbf{|Difference Sets|} &    \textbf{Algorithm}    & \textbf{\diffonly{}} & \textbf{\scratch{}} \\ \midrule
        \multirow{2}{*}{\textbf{1K}}  & \textbf{BFS}  & \bf{1.4s}       & 13.5s             \\
                                         & \textbf{PR} &      \textbf{66.5s}   &    136.2s          \\ \midrule
        \multirow{2}{*}{\textbf{3.5M}} & \textbf{BFS}  & \bf{13.0s}       & 25.7s             \\
                                         & \textbf{PR} & 281.9s        & \bf{193.2s}             \\ \bottomrule
        \end{tabular}
        \vspace{3pt}
        \caption{Runtimes of BFS and PR for two view collections on the Orkut graph, containing 1K- and 3.5M-size difference sets, in two ways: (i) \diffonly{}; and (ii) \scratch{}.}\label{table:splitting-motivation}
        \vspace{-25pt}
\end{table}

Several factors determine how big the difference is between $A$'s footprint on two consecutive views $GV_i$ and $GV_{i+1}$, which determines how expensive it is to compute $GV_{i+1}$ differentially. Two of these factors can be observed by \graphsurge{}: (1) how stable is $A$? (2) how large are the difference sets?  
We use the term unstable to refer to computations that may generate a lot of differences on small differences to input datasets. We next demonstrate these factors through a controlled experiment. We will also demonstrate this on a more realistic application in Section~\ref{sec:experiments}. We take 10M edges from the Orkut social network graph and construct an initial view $GV_1$ and then construct two artificial view collections each containing 20 views: (i) $C_{1K}$, in which we randomly add 500 edges and remove 500 edges to each $GV_{i-1}$; (ii) $C_{3.5M}$, in which we add 2M edges and remove 1.5M edges randomly to each $GV_{i-1}$. The sizes of the difference sets are picked to obtain a collection with highly similar and highly different views, respectively. We then run Breadth First Search (BFS) and PageRank (PR) on both collections in two ways: (i)~\diffonly{}: runs the collection only differentially; and (ii)~\scratch{}: runs each view in the collection from scratch.

Table~\ref{table:splitting-motivation} shows the runtimes. First, notice that on $C_{3.5M}$, while it is better to run BFS differentially, it is better to run PageRank from scratch. This is because PageRank is a less stable algorithm than BFS. For example, assume $GV_{i+1} = GV_i \cup \{u$$\rightarrow$$v\}$, so the views differ by a single edge addition, and consider differentially fixing the first iteration of BFS\@. 
This addition results in 1 difference in the \texttt{JoinMsg} operator. 
In vertex-centric terms, it will result in $u$ sending 1 more extra message to $v$ containing $u$'s current distance. 
 In contrast, in PageRank, $u$ sends a message of $1/deg(u)$ to its neighbors so all messages that $u$ sends might change. 
Second, observe that when the views are sufficiently similar, specifically 
in $C_{1K}$, running PageRank differentially also starts to be the better option. That is, the size of the differences also determines whether running views differentially vs from scratch is the better option.

We have implemented an adaptive optimizer that decides whether to run each view $GV_i$ in a view collection differentially or from scratch. 
Our optimizer observes two simple runtime metrics to make its splitting decisions: (1)~Each time the system decides to split the collection at $GV_{i}$ and run $GV_{i}$ from scratch, we measure how long it took to compute $A$ on $GV_{i}$ from scratch and what was the size of $G_{i}$ and (2)~Each time the system decides to run a $GV_{i}$ differentially we keep track of how long it took to run $GV_{i}$ differentially and what was the size of $\delta C_i$. Then, for each $GV_i$, we use two simple linear models to estimate how long it would take to rerun $GV_i$ from scratch and differentially given, respectively, the sizes of $GV_i$ and $\delta C_i$, and pick the faster estimated option. Specifically:
\begin{squishedenumerate}
\item[1.] Run $GV_1$ from scratch and $GV_2$ differentially and keep track of ($|GV_1|$, $st_1$), for {\bf s}cratch {\bf t}ime, and ($\delta C_2$, $dt_2$), for {\bf d}ifferential {\bf t}ime. 

\item[2.] Then for each other view $GV_i$ for $i=3,\ldots,k$, estimate the run time of running $GV_i$ from scratch or differentially using the collected $st_j$ times and the size of $|GV_i|$ and the $dt_j$ points and the size of $\delta C_i$. 
\end{squishedenumerate}

\noindent 
In our actual implementation, we make splitting decisions for $\ell$ views at a time (10 by default) as feeding multiple views to DD makes DD's data indexing code run faster.
\noindent We will demonstrate that our optimizer is both able to adapt to running computations differentially or from scratch, when either option is superior, and can even outperform both options in some cases by selectively splitting collections in a subset of the views.

We next discuss an important question: How much faster can an algorithm $A$ running differentially on a view collection $C$  be compared to running $A$ on each view from scratch (and vice versa)? 
A high-level answer should instruct the benefits we can expect from adaptive splitting.
Consider a $k$-view collection $C$, where each view is identical. This is conceptually the best case for running $A$ on $C$ differentially, where after the first view, the rest of the  views are computed instantaneously. Therefore, differentially computing $A$ can be {\em $k$ factor better} than running $A$ from scratch. 
Interestingly, the situation is not similar in the reverse direction. The worst case for running $A$ on $C$ differentially is if each view was completely disjoint, i.e., $\delta C_{i}$$=$\{$-GV_{i-1}$$\cup$$+GV_i$\}. We effectively completely remove $GV_{i-1}$ and add $GV_{i}$. Therefore, when running $A$ on $GV_i$ differentially, $DD$, to the first approximation, will ``undo'' computation for $GV_{i-1}$ and then run $A$ from scratch differentially. We effectively compute $A$ on each view twice, and should expect a bounded, around 2x, slow down to running computations differentially even in this worst case. This is an important robustness property of running computations differentially.  It is still important to perform our splitting optimization because: (i) there can still be a significant performance gain over pure differential computation (we will report up to 1.9x improvements); and (ii) some unstable computations consistently perform better when computed from scratch and our splitting optimization automatically detects those cases.

\vspace{-5pt}
\section{Evaluation}\label{sec:experiments}
We next present our experiments. Section~\ref{sec:exp-diff-scratch} starts by empirically demonstrating the possible performance gains of running computations differentially across views vs running them from scratch. Section~\ref{sec:exp-collection-splitting} and~\ref{sec:exp-ordering}, respectively, evaluate the benefits of our collection splitting and ordering optimizations. Section~\ref{sec:exp-baseline} presents baseline comparisons between $DD$ and the GraphBolt~\cite{mariappan:graphbolt}. Finally,  Section~\ref{sec:exp-scalability}, 
presents that \graphsurge{} obtains good scalability across compute nodes in a cluster. 

\vspace{-5pt}
\subsection{Experimental Setup}

\customsection{Datasets:} We evaluate \graphsurge{} on 5 real-world graphs. \emph{size} below indicates the size of each dataset on disk.

\begin{squishedlist}
    \item \textbf{Stack Overflow}~\cite{snap-sxso} (SO, |V| = 2M, |E| = 63M, size = 1.6GB) is a temporal dataset where every edge has an associated UNIX timestamp indicating its creation time.
    \item \textbf{Paper Citations} (PC, |V| = 172M, |E| = 605M, size = 14.8GB) is a paper-to-paper citation graph constructed from the Semantic Scholar Corpus~\cite{semanticscholar} (version 2019-10-01). The vertices have 2~associated properties: the year of publication and the count of co-authors.
    \item \textbf{Com-Livejournal}~\cite{snap-comlj} (CLJ, |V| = 4M, |E| = 34M, size = 1.1GB) is a social network graph containing a list of ground-truth communities representing social groups that a subset of the users are part of. Users can be part of multiple communities.
    \item \textbf{Com-Wiki-Topcats}~\cite{snap-wtc} (WTC, |V| = 1M, |E| = 28M, size = 719MB) is a web graph whose vertices can belong to one or more communities representing the category of a web page.
    \item \textbf{Twitter}~\cite{kwak:twitter} (TW, |V| = 42M, |E| = 1.5B, size = 25GB) is a social network graph.
\end{squishedlist}

\pgfplotstableread[col sep=tab]{
timeperiods	diff	scratch	adapt	comparex	empty
1d	4.89	57.12	3.40	11.7x	
1m	7.63	38.80	7.38	5.1x	
6m	8.07	21.34	7.39	2.6x	
1yr	7.83	16.56	7.44	2.1x	
2yr	6.85	14.50	6.64	2.1x	
1d	2.26	30.81	2.23	13.6x	
1m	2.86	21.29	3.01	7.4x	
6m	3.35	12.18	3.24	3.6x	
1yr	4.07	9.27	3.07	2.3x	
2yr	3.11	7.01	3.12	2.3x	
1d	44.09	431.03	43.23	9.8x	
1m	85.56	312.10	90.57	3.6x	
6m	68.93	156.01	70.33	2.3x	
1yr	61.12	115.99	60.13	1.9x	
2yr	57.07	87.07	57.19	1.5x	
1d	109.59	154.15	100.51	1.4x	
1m	94.76	112.44	104.28	1.2x	
6m	48.88	57.26	48.88	1.2x	
1yr	36.23	43.19	36.31	1.2x	
2yr	29.43	30.71	28.78	1.0x	
1d	2.85	30.60	2.78	10.7x	
1m	3.03	21.94	3.26	7.2x	
6m	3.54	11.46	3.37	3.2x	
1yr	3.17	9.18	3.69	2.9x	
2yr	3.16	6.98	3.01	2.2x	
1d	7.60	105.84	7.19	13.9x	
1m	15.93	84.49	17.15	5.3x	
6m	18.18	43.44	17.70	2.4x	
1yr	15.73	32.56	15.77	2.1x	
2yr	16.51	24.83	16.08	1.5x	
}{\csim}

\pgfplotstableread[col sep=tab]{
timeperiods	diff	scratch	adapt	comparex	empty
6m	17.57	16.44	17.96	1.1x	
1yr	15.76	15.73	16.53	1.0x	
2yr	16.24	15.60	13.59	1.0x	
3yr	10.84	9.65	11.46	1.1x	
4yr	15.41	13.78	14.34	1.1x	
6m	7.27	10.43	10.80	0.7x	
1yr	8.89	8.06	8.17	1.1x	
2yr	8.82	6.96	7.30	1.3x	
3yr	5.59	4.78	5.68	1.2x	
4yr	6.97	6.96	6.85	1.0x	
6m	142.06	57.99	61.03	2.4x	
1yr	121.18	48.01	47.94	2.5x	
2yr	98.88	43.58	45.13	2.3x	
3yr	70.69	45.77	46.34	1.5x	
4yr	92.23	51.39	50.77	1.8x	
6m	40.19	26.51	27.78	1.5x	
1yr	36.23	22.31	22.76	1.6x	
2yr	34.00	21.45	21.71	1.6x	
3yr	24.27	17.58	16.29	1.4x	
4yr	30.70	19.68	19.76	1.6x	
6m	7.49	10.82	11.25	0.7x	
1yr	8.58	8.59	6.97	1.0x	
2yr	9.22	8.05	7.82	1.1x	
3yr	5.72	5.95	5.92	1.0x	
4yr	7.11	6.59	7.42	1.1x	
6m	19.77	17.18	16.99	1.2x	
1yr	22.03	15.09	15.67	1.5x	
2yr	21.43	15.08	15.42	1.4x	
3yr	18.18	13.61	13.37	1.3x	
4yr	21.64	13.71	15.06	1.6x	
}{\cno}

\begin{figure*}[t]
    \centering
    \begin{subfigure}{\linewidth}
        \centering
        \begin{tikzpicture}
            \begin{groupplot}[
                    group style={
                        group name=group,
                        group size=6 by 1,
                        xlabels at=edge bottom,
                        ylabels at=edge left,
                        horizontal sep=0.55cm,
                    },
                    width=3.85cm,
                    height=4cm,
                    /pgf/bar width=0.08cm,
                    ybar=0,
                    ylabel={(a) $\boldsymbol{C_{sim}}$\\Runtime (s)},
                    ylabel style={align=center},
                    ylabel shift=-5pt,
                    xtick=data,
                    xticklabels from table={\csim}{timeperiods},
                    xtick pos=left,
                    ytick pos=left,
                    yticklabel style={
                        font=\scriptsize,
                        /pgf/number format/set thousands separator={}
                    },
                    xticklabel style={
                        font=\small
                    },
                    enlarge x limits=0.15,
                    ymin=0,
                    cycle list name=three-colors,
                    nodes near coords,
                    nodes near coords style={
                        rotate=90,
                        anchor=west,
                        font=\small,
                        color=black
                    },
                    point meta=explicit symbolic,
                    legend entries={diff,scratch,adapt},
                    legend cell align=left,
                    legend to name=grouplegendds,
                    legend style={
                        draw=none,
                        fill=none,
                        text opacity = 1,
                        row sep=-3.5pt,
                        font=\footnotesize
                    },
                ]
                \nextgroupplot[xlabel=WCC,enlarge y limits={upper, value=0.5}]
                \foreach \y/\meta in {1/empty,2/comparex,3/empty} {
                    \addplot +[restrict expr to domain={\coordindex}{0:4}] table [
                        x expr=\coordindex,
                        y index=\y,
                        meta=\meta,
                    ] {\csim};
                }
                \nextgroupplot[xlabel=BFS,enlarge y limits={upper, value=0.5}]
                \foreach \y/\meta in {1/empty,2/comparex,3/empty} {
                    \addplot +[restrict expr to domain={\coordindex}{5:9}] table [
                        x expr=\coordindex,
                        y index=\y,
                        meta=\meta,
                    ] {\csim};
                }
                \nextgroupplot[xlabel=SCC,enlarge y limits={upper, value=0.5}]
                \foreach \y/\meta in {1/empty,2/comparex,3/empty} {
                    \addplot +[restrict expr to domain={\coordindex}{10:14}] table [
                        x expr=\coordindex,
                        y index=\y,
                        meta=\meta,
                    ] {\csim};
                }
                \nextgroupplot[xlabel=PR,enlarge y limits={upper, value=0.5}]
                \foreach \y/\meta in {1/empty,2/comparex,3/empty} {
                    \addplot +[restrict expr to domain={\coordindex}{15:19}] table [
                        x expr=\coordindex,
                        y index=\y,
                        meta=\meta,
                    ] {\csim};
                }
                \nextgroupplot[xlabel=SSSP,enlarge y limits={upper, value=0.5}]
                \foreach \y/\meta in {1/empty,2/comparex,3/empty} {
                    \addplot +[restrict expr to domain={\coordindex}{20:24}] table [
                        x expr=\coordindex,
                        y index=\y,
                        meta=\meta,
                    ] {\csim};
                }
                \nextgroupplot[xlabel=MPSP,enlarge y limits={upper, value=0.5}]
                \foreach \y/\meta in {1/empty,2/comparex,3/empty} {
                    \addplot +[restrict expr to domain={\coordindex}{25:29}] table [
                        x expr=\coordindex,
                        y index=\y,
                        meta=\meta,
                    ] {\csim};
                }
            \end{groupplot}
            \node[anchor= north east] (leg) at ($(group c1r1.north east) + (0.2cm,0.15cm)$){\pgfplotslegendfromname{grouplegendds}};
        \end{tikzpicture}
    \end{subfigure}
    \begin{subfigure}{\linewidth}
        \centering
        \begin{tikzpicture}
            \begin{groupplot}[
                    group style={
                        group name=group,
                        group size=6 by 1,
                        xlabels at=edge bottom,
                        ylabels at=edge left,
                        horizontal sep=0.5cm,
                    },
                    width=3.85cm,
                    height=4cm,
                    /pgf/bar width=0.09cm,
                    ybar=0,
                    ylabel={(b) $\boldsymbol{C_{no}}$\\Runtime (s)},
                    ylabel style={align=center},
                    ylabel shift=-5pt,
                    xtick=data,
                    xticklabels from table={\cno}{timeperiods},
                    xtick pos=left,
                    ytick pos=left,
                    yticklabel style={
                        font=\scriptsize
                    },
                    xticklabel style={
                        font=\small
                    },
                    enlarge x limits=0.15,
                    ymin=0,
                    cycle list name=three-colors,
                    nodes near coords,
                    nodes near coords style={
                        rotate=90,
                        anchor=west,
                        font=\small,
                        color=black
                    },
                    point meta=explicit symbolic,
                ]
                \nextgroupplot[xlabel=WCC,enlarge y limits={upper, value=0.5}]
                \foreach \y/\meta in {1/empty,2/empty,3/comparex} {
                    \addplot +[restrict expr to domain={\coordindex}{0:4}] table [
                        x expr=\coordindex,
                        y index=\y,
                        meta=\meta,
                    ] {\cno};
                }
                \nextgroupplot[xlabel=BFS,enlarge y limits={upper, value=0.5}]
                \foreach \y/\meta in {1/comparex,2/empty,3/empty} {
                    \addplot +[restrict expr to domain={\coordindex}{5:9}] table [
                        x expr=\coordindex,
                        y index=\y,
                        meta=\meta,
                    ] {\cno};
                }
                \nextgroupplot[xlabel=SCC,enlarge y limits={upper, value=0.5}]
                \foreach \y/\meta in {1/comparex,2/empty,3/empty} {
                    \addplot +[restrict expr to domain={\coordindex}{10:14}] table [
                        x expr=\coordindex,
                        y index=\y,
                        meta=\meta,
                    ] {\cno};
                }
                \nextgroupplot[xlabel=PR,enlarge y limits={upper, value=0.5}]
                \foreach \y/\meta in {1/comparex,2/empty,3/empty} {
                    \addplot +[restrict expr to domain={\coordindex}{15:19}] table [
                        x expr=\coordindex,
                        y index=\y,
                        meta=\meta,
                    ] {\cno};
                }
                \nextgroupplot[xlabel=SSSP,enlarge y limits={upper, value=0.5}]
                \foreach \y/\meta in {1/comparex,2/empty,3/empty} {
                    \addplot +[restrict expr to domain={\coordindex}{20:24}] table [
                        x expr=\coordindex,
                        y index=\y,
                        meta=\meta,
                    ] {\cno};
                }
                \nextgroupplot[xlabel=MPSP,enlarge y limits={upper, value=0.5}]
                \foreach \y/\meta in {1/comparex,2/empty,3/empty} {
                    \addplot +[restrict expr to domain={\coordindex}{25:29}] table [
                        x expr=\coordindex,
                        y index=\y,
                        meta=\meta,
                    ] {\cno};
                }
            \end{groupplot}
        \end{tikzpicture}
    \end{subfigure}%
    \vspace{-10pt}
    \caption{Runtime of algorithms showing benefits of running each view: (a) differentially or (b) from scratch.}%
    \label{fig:exp-diffs}
    \vspace{-10pt}
\end{figure*}
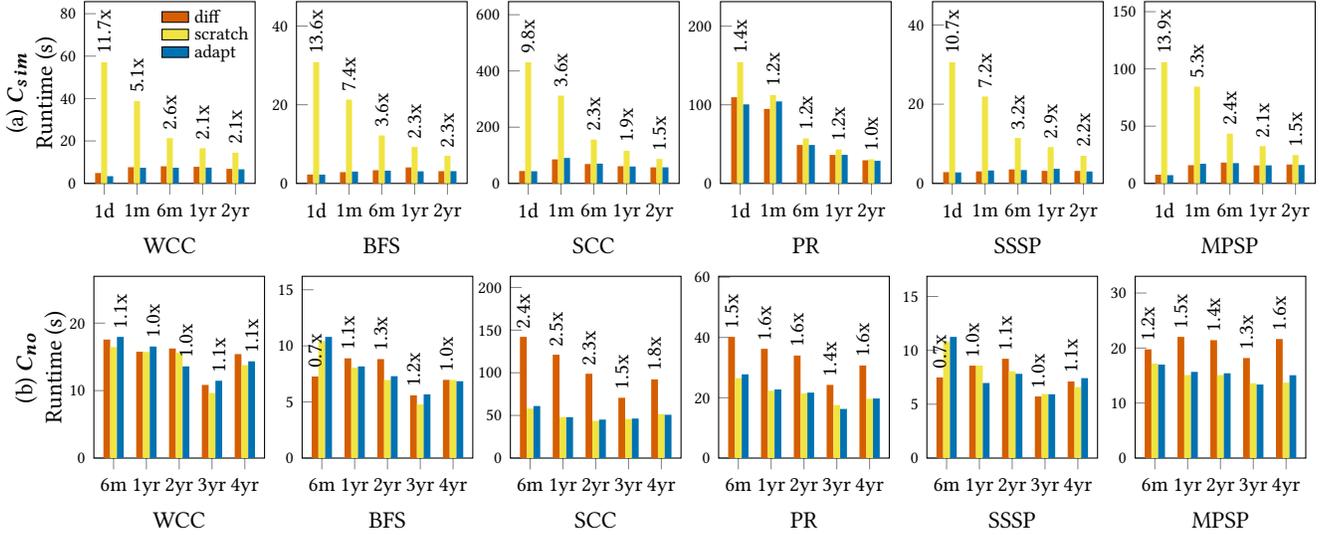

\customsection{Computations:} We use 6~different analytics computations: (i) weakly connected components (WCC); (ii) strongly connected components (SCC), which implements the doubly-iterative Coloring algorithm~\cite{orzan:thesis}; (iii) breadth-first search (BFS); (iv) single source shortest path (SSSP); (v) PageRank (PR); and (vi) multiple pair shortest path (MPSP). For BFS and SSSP, we set the source to a random vertex that has outgoing edges. For MPSP, we randomly select 5~pairs of vertices (src,dst), where src is a vertex with outgoing edges and dst is a vertex with incoming edges. All computations are implemented using \graphsurge{}'s DD-based analytics API.

\customsection{Hardware and Software:} We compiled \graphsurge{} using \texttt{rustc} v1.46.0, \texttt{timely-dataflow} v0.11.0, and \texttt{differential-dataflow} v0.11.0 and performed our experiments on a cluster of up to 12 machines each running Ubuntu 18.04.3. Each machine has 2x Intel E5-2670 @2.6GHz CPU with 32 logical cores. Every machine has 256 GB RAM, except 2, which have 512GB RAM. Except our scalability experiments, all experiments were performed on a single machine.

\vspace{-5pt}
\subsection{Comparison of Differential Computing vs Rerunning from Scratch}%
\label{sec:exp-diff-scratch}

Recall our observation from Section~\ref{sec:collection-splitting} that while differentially computing $A$ can be unboundedly faster than running from scratch, the reverse comparison is bounded.  We start by demonstrating this intuition empirically. We model a historical analysis application, where we build two sets of  view collections on the SO dataset:
\begin{squishedenumerate}
\item C$_{sim}$: are a set of {\bf s}imilar view collections that each starts with a 5-year window of the graph, from  May 2008 to May 2013, which forms the first view. Then we set a time window of size~$w$ of 1 day, 1 month, 6 moths, 1 year, and 2 years, and {\em expand} the initial window by~$w$, so each view $GV_i$ includes $GV_{i-1}$ plus an additional number of edges for a larger $w$-size window. This generates 5 collections. C$_{sim, 1d}$, where $w$ is 1 day, contains the most similar and largest number of views. C$_{sim, 2y}$ is the least similar and contains the fewest number of views. 
\item C$_{no}$: are a set of {\bf n}on-{\bf o}verlapping, so highly different views, where we start with a window of the graph from May 2008 till December 2008, then we completely {\em slide} the window by a window of size~$w$ of 6 months, 1, 2, 3, and 4 years. This generates~5 collections, all of which are completely non-overlapping. The window size $w$ allows us to create collections with increasingly more views. 
\end{squishedenumerate}
We evaluate the performance of 6 algorithms on each collection, turning our splitting and ordering optimizers off, in two ways:  \diffonly{} and \scratch{}, which were described in Section~\ref{sec:collection-splitting}.
We expect \diffonly{} to be more performant than \scratch{} in each C$_{sim}$ collection, but increasingly more as $w$ gets smaller and there are a larger number of views. We expect \scratch{} to be more performant in each C$_{no}$ collection, but we do not expect to see increasingly more gains as the number of views increases. Figures~\ref{fig:exp-diffs} show our results for the C$_{sim}$ and C$_{no}$ collections, respectively.
Observe that in C$_{sim}$ collections, indeed as $w$ gets smaller, we see an increasing factor on benefits  for \diffonly{} varying from 1.5x to 13.9x. The only exception is PageRank, which we observed is not as stable as the rest of our algorithms. In contrast, 
in the C$_{no}$ collection, we see up to 2.5x 
performance improvements for \scratch{}, 
but we do not observe improved factors with increasing number of views. %

\subsection{Benefits of Collection Splitting}%
\label{sec:exp-collection-splitting}

We next evaluate \graphsurge{}'s adaptive splitting optimizer continuing our previous set up. We refer to this configuration as \texttt{adaptive}. We still keep our ordering optimizer off to only study the behavior of our adaptive optimizer, which we refer to as \adaptive{}.
We reran the previous experiment with \adaptive{}. The \adaptive{} bar in Figure~\ref{fig:exp-diffs} shows our results.
Except for two experiments, running BFS and SSSP on $C_{no}$ with 6 month slides, 
\adaptive{} is able to perform as good or almost as good as the better of \diffonly{} or \scratch{}. Note that in these experiments, it is always better to either run the computations with one of \diffonly{} or \scratch{}, so we do not expect \adaptive{} to outperform both of these strategies. Importantly, in almost all cases, we adapt to the better strategy. 

Next we created a view collection in which \adaptive{} can outperform both \diffonly{} and \scratch{}. Specifically, we created a view collection C$_{\text{aut}}$ on the PC citation dataset, which contains the Cartesian product of two sets of windows on two properties. First is a 5 year non-overlapping window from $[1996, 2000]$ to $[2016, 2020]$. The other is a window for the number of \textbf{aut}hors on the papers, that expands from $[0, 5]$ to $[0, 25]$ in windows of size 5. For example, the view $[1996, 2000]x[0, 5]$ is the view that contains all papers written between 1996 and 2000 containing at most 5 authors and their citations. This collection contains views that generates a sequence of addition-only differences as the number of authors window expands, and then a non-overlapping view, when the year window slides, creating a potential splitting point.
\noindent Table~\ref{table:adaptive-citation} shows the runtimes of 6 algorithms on C$_{\text{aut}}$. Observe that \adaptive{} matches or outperforms, by up to 1.9x, the better of \diffonly{} and \scratch{}. \adaptive{} is able to pick the splitting points where the year window slides and consistently outperforms \diffonly{} and \scratch{} when running all algorithms.

\begin{table}[t]
    \centering
    \begin{tabular}{@{}lrlrlrl@{}}
        \toprule
                & \multicolumn{1}{l}{WCC} &               & \multicolumn{1}{c}{BFS}  &               & \multicolumn{1}{c}{SCC}  &               \\ \midrule
        diff    & 117.26                  & (1.9$\times$) & 19.29                    & (1.4$\times$) & 314.78                   & (1.8$\times$) \\
        scratch & 120.53                  & (1.9$\times$) & 44.20                    & (3.3$\times$) & 351.83                   & (2.0$\times$) \\
        adapt   & 61.88                   &               & 13.54                    &               & 179.27                   &               \\ \midrule
                & \multicolumn{1}{c}{PR}  &               & \multicolumn{1}{c}{SSSP} &               & \multicolumn{1}{c}{MPSP} &               \\ \midrule
        diff    & 79.07                   & (1.6$\times$) & 18.6607                  & (1.3$\times$) & 20.4278                  & (1.2$\times$) \\
        scratch & 114.74                  & (2.3$\times$) & 42.3927                  & (3.0$\times$) & 43.7398                  & (2.6$\times$) \\
        adapt   & 50.13                   &               & 14.347                   &               & 16.6355                  &               \\ \bottomrule
        \end{tabular}
    \caption{Runtime (seconds) of algorithms for the C$_{\text{aut}}$ collection showing that the adaptive optimizer can outperform both running differentially and from scratch.}%
    \label{table:adaptive-citation}
    \vspace{-25pt}
\end{table}

\subsection{Benefits of Collection Ordering}\label{sec:exp-ordering}

The goal of our next experiment is to study the performance gains of our collection ordering optimization. We develop a perturbation analysis application on our graph with ground truth communities, namely CLJ\@. We construct view collections by taking the largest N~communities and remove each k~combination of these N~communities to perturb the graphs in a variety of ways. Specifically we construct two collections for two N, k combinations:  C$_{10, 5}$ sets N=10 and k=5 and contains 252 views, and C$_{7, 4}$ sets N=7 and k=4 and contains 35 views. Note that this is an application where finding a good manual order is difficult, as each view removes possibly millions of edges, and there are hundreds of views in the collection. Therefore as a baseline, we will use random collection orderings. 

We first turned our adaptive splitting optimizer off to isolate the benefits due to collection ordering only and compared the performance of the order that \graphsurge{} picks, which we call \texttt{Ord},
with one random ordering, which we call \texttt{R}. Our experiments had two more random orderings, which behave almost the same as \texttt{R}, but we omit those numbers due to space constraints.
We then executed 6 algorithms on the view collections.
The \texttt{no adapt} bars in Figure~\ref{fig:exp-ordering-adaptive} show our results. Table~\ref{table:exp-ordering} presents the amount of total edge differences in our edge difference sets. Observe that: (i)~our optimizer's order generates between 3.4x to 16.8x fewer differences than the random order; and (ii)~our ordering optimization improves performance consistently and between 1.3x to 9.8x across our experiments. 
For reference, Table~\ref{table:exp-ordering} also reports the times it takes \graphsurge{} to compute the collections, with and without ordering, in row CCT (collection creation time). The difference between the random order's CCT and \texttt{Ord}'s CCT is the overhead of ordering, which ranged between 1.3x and 1.8x.

\begin{table}[t]
    \centering
    \begin{tabular}{@{}ccr@{\hskip 4pt}l@{\hskip 4pt}r@{\hskip 5pt}l@{\hskip 4pt}l@{\hskip 4pt}r@{\hskip 5pt}l@{}}
        \toprule
        \multicolumn{1}{l}{} & \multicolumn{1}{l}{} & \multicolumn{1}{l}{} &  & \multicolumn{2}{c}{Ord}  &  & \multicolumn{2}{c}{R}         \\ \cmidrule(r){1-3} \cmidrule(lr){5-6} \cmidrule(l){8-9} 
        CLJ                  & 10C5                 & \# Diffs             &  & \textbf{158M} &          &  & \textbf{1.5B} & (9.6$\times$)  \\
                            &                      & CCT                  &  & 355.0         & (+151.5) &  & 203.5         & (1.7$\times$)  \\ \cmidrule(l){2-9} 
                            & 7C4                  & \# Diffs             &  & \textbf{54M}  &          &  & \textbf{191M} & (3.6$\times$)  \\
                            &                      & CCT                  &  & 38.8          & (+8.7)   &  & 30.2          & (1.3$\times$)  \\ \midrule
        WTC                  & 10C5                 & \# Diffs             &  & \textbf{73M}  &          &  & \textbf{1.2B} & (16.8$\times$) \\
                            &                      & CCT                  &  & 299.0         & (+128.5) &  & 170.5         & (1.8$\times$)  \\ \cmidrule(l){2-9} 
                            & 7C4                  & \# Diffs             &  & \textbf{44M}  &          &  & \textbf{149M} & (3.4$\times$)  \\
                            &                      & CCT                  &  & 35.1          & (+8.5)   &  & 26.6          & (1.3$\times$)  \\ \bottomrule
        \end{tabular}
    \caption{The number of diffs and collection creation time (CCT) in seconds for C$_{10,5}$ and C$_{7,4}$  on CLJ and WTC for a random order \texttt{R1} and our optimizer's order.}%
    \label{table:exp-ordering}
    \vspace{-30pt}
\end{table}

\input{table-ordering}

We next turned the adaptive splitting optimization on to measure the performance benefits in the full system. 
This forms a full end-to-end experiment as both of our optimizations are turned on.
We expect the benefits of ordering to decrease when adaptive splitting actually splits and improves the performance 
of the random ordering. If adapting defaults to running only differentially, we expect the results to be similar to our previous results.
The \texttt{adapt} bars in Figure~\ref{fig:exp-ordering-adaptive} show our results. 
There are 3 experiments in which adapting improves the random order's performance by splitting: when running SCC on CLJ with both  C$_{10, 5}$ and 
C$_{7, 4}$ collections and running SSSP on CLJ with C$_{7, 4}$ experiment. In these cases, the benefits of ordering
decreases compared to when adapting optimization was off. For example, when running SCC
on  CLJ with  C$_{10, 5}$ collection, the benefits of ordering decreases from 4.1x to 2.0x. In other experiments,
our adaptive optimization defaults to running all computations differentially (or performs
 slightly worse than running only differentially). In these cases, the ordering optimization improves the performance
 similar to when adaptive optimization was off (between 1.1x to 10.1x).

\subsection{Baseline Temporal Systems}\label{sec:exp-baseline}

In this section, we provide baseline comparisons against GraphBolt~\cite{mariappan:graphbolt} (GB). 
GB is a shared-memory streaming system that is developed on top of Ligra~\cite{shun:ligra} and designed to maintain computation results over a stream of updates.  As such, we can develop a \graphsurge{}-like system on top of GB by feeding our view collections as an evolving graph to GB instead to DD.
The primary difference between DD and GB, and the reason we chose DD, is that GB requires users to write explicit maintenance code in functions such as \texttt{retract} or \texttt{propagatedelta}, which is challenging for some algorithms, such as the doubly-iterative SCC algorithm. 

We evaluate the performance of \graphsurge{} and GB for two computations, SSSP and PR, on the TW datset. We simulate a temporal analytics application in \graphsurge{} by constructing a view collection with 1001 views, where the first view contains 50\% of the total edges in the original graph selected randomly, and each of the remaining views contains 500 additions and 500 deletions based on the previous view. Figure~\ref{fig:baseline-new} shows the comparison results. We note that \graphsurge{} is up to 6.4$ \times$ faster than GB for SSSP, but is up to 13.5$\times$ slower than GB for PR\@. These numbers are similar to the numbers in reference~\cite{mariappan:graphbolt} (Figures 8 and 9) for GB and DD\@.

There are two primary reasons for the lower performance of \graphsurge{} for PR\@.
First, DD's execution engine uses a dataflow architecture, which is based on message passing, which has a higher runtime overhead as compared to the shared memory architecture of GB, which propagates updates by directly writing to memory locations using atomic operations. Despite this advantage, as reference~\cite{mariappan:graphbolt}, we found DD to be more performant on SSSP.
Second, \graphsurge{} uses DD as its analytics engine, which is built to support general incremental computation. However, this generality can naturally come at a performance cost, because DD is unable to take advantage of computation-specific optimizations. For example, GB implements a PR-specific incrementalization code, which is more efficient than differential computation. However, specialized incremental versions of many algorithms is very similar to differential computation. For example, GB's incremental SSSP algorithm~\cite{mariappan:graphbolt} is effectively differential computation and for such algorithms, DD generates equally efficient incremental versions automatically.

\pgfplotstableread[col sep=tab]{
graph	GS	gsx	GB	gbx
	174.15		1108.37	6.4x
	34768.52	2.3x	14914.52	
}{\gsgb}

\pgfplotstableread[col sep=tab]{
x	WCC	BFS	SSSP	PR
1	493.39	260.81	458.14	1113.92
2	320.55	210.40	323.27	653.81
4	173.06	95.06	161.37	324.03
8	148.80	55.75	107.39	203.36
12	139.05	43.05	93.72	178.21
}{\distributed}

\begin{figure}[t]
    \centering
    \begin{minipage}[b][][b]{0.47\columnwidth}
        \begin{tikzpicture}
            \begin{groupplot}[
                    group style={
                        group name=group,
                        group size=2 by 1,
                        xlabels at=edge bottom,
                        ylabels at=edge left,
                    },
                    width=2.5cm,
                    height=4.5cm,
                    /pgf/bar width=0.3cm,
                    ybar=0,
                    ylabel={Runtime (seconds)},
                    ylabel shift=-5pt,
                    xlabel shift=-5pt,
                    xtick=data,
                    xticklabels from table={\gsgb}{graph},
                    xtick pos=left,
                    ytick pos=left,
                    yticklabel style={
                        /pgf/number format/set thousands separator={}
                    },
                    enlarge x limits=0.3,
                    ymin=0,
                    cycle list name=gs-gb,
                    nodes near coords,
                    nodes near coords style={
                        rotate=90,
                        anchor=west,
                        color=black
                    },
                    point meta=explicit symbolic,
                    legend columns=2,
                    legend entries={GS,GB},
                    legend cell align=left,
                    legend to name=grouplegendgsb,
                    legend style={
                        draw=none,
                        fill=none,
                        text opacity = 1,
                        row sep=-4.5pt,
                    },
                ]
                \nextgroupplot[xlabel=SSSP,enlarge y limits={upper, value=0.5}]
                \foreach \y/\meta in {1/gsx,3/gbx} {
                    \addplot +[restrict expr to domain={\coordindex}{0:0}] table [
                        x expr=\coordindex,
                        y index=\y,
                        meta=\meta,
                    ] {\gsgb};
                }
                \nextgroupplot[xlabel=PR,enlarge y limits={upper, value=0.5}]
                \foreach \y/\meta in {1/gsx,3/gbx} {
                    \addplot +[restrict expr to domain={\coordindex}{1:1}] table [
                        x expr=\coordindex,
                        y index=\y,
                        meta=\meta,
                    ] {\gsgb};
                }
            \end{groupplot}
            \node[anchor= south west] (leg) at ($(group c1r1.north west) + (-0.2cm,-0.1cm)$){\pgfplotslegendfromname{grouplegendgsb}};
        \end{tikzpicture}
        \vspace{-15pt}
        \caption{\graphsurge{} vs GraphBolt.}%
        \label{fig:baseline-new}
    \end{minipage}
    \hfill
    \begin{minipage}[b][][b]{0.47\columnwidth}
        \begin{tikzpicture}
            \begin{axis}[
                height=4.5cm,
                width=4.5cm,
                ymin=0,
                enlarge y limits={upper, value=0.1},
                xtick=data,
                xticklabels from table={\distributed}{x},
                xlabel={\# compute machines},
                ylabel={},
                cycle list name=exotic,
                legend entries={SSSP,PR},
                legend style={
                    draw=none,
                    fill=none,
                    text opacity = 1,
                    row sep=-3pt,
                    font=\small
                },
            ]
            \foreach \y in {SSSP,PR} {
                \addplot table [
                    x=x,
                    y=\y,
                ] {\distributed};
            }
            \end{axis}
        \end{tikzpicture}
        \vspace{-5pt}
        \caption{Scaling runtime in a distributed setting.}%
        \label{fig:scaling-new}
    \end{minipage}
    \vspace{-10pt}
  \end{figure}

\vspace{-5pt}
\subsection{Distributed Execution and Scalability}%
\label{sec:exp-scalability}

We next demonstrate the ability of \graphsurge{} to scale in a distributed setting. We modified the Twitter dataset by assigning artificial city, state, and country attributes to the vertices and an affinity weight to edges that indicates the level of interaction between users. We modeled a social network analysis application that studies the connected users who live within the same city, state, and country with three different affinity levels, low, medium, and high, constructing a view collection with 9 views. We measured the runtime for 2 algorithms: SSSP and PR, on this view collection using up to 12 compute machines, each with 32 worker threads. Figure~\ref{fig:scaling-new} shows the scalability results on this large view collection. Additional machines improve the runtime for both of the algorithms almost linearly. This experiment demonstrates that \graphsurge{} is able to take full advantage of TD and DD for seamlessly scaling to multiple machines in a distributed environment.

\vspace{-5pt}
\section{Related Work}\label{sec:related-work}

\customsection{Incremental view maintenance and computation sharing ac\-ross queries:}
Differential computation is a technique to maintain the outputs
of arbitrary dataflow computations as their inputs change. 
\graphsurge{} leverages differential computation for computation sharing. There is an extensive
literature on incremental view maintenance in database literature.
These techniques focus on maintaining the outputs of relational or
datalog queries. We refer the reader to the following references on these topics~\cite{widom:incremental, gupta:ivm}. 
Work on incremental graph computations will be covered in more detail later in this section.

Differential computation is unique as a maintenance technique because it 
stores the computational footprint
of a computation $A$ on an input $E$ and detects  
and shares this stored computation when $A$ is executed on an updated version
$E'$ of the input. Although different, this is similar to work that shares 
computations across multiple queries
that run over a single input by detecting common sub-expressions.
Examples include work on continuous querying system, such as NiagaraCQ~\cite{ChenJDTW00}, 
or systems that support running multiple queries, e.g., when incrementally maintaining
multiple views~\cite{KaranasosKM13} or running multiple queries in a batch~\cite{roy:mqo}.

\customsection{Fan et al.~\cite{fan:incremental}} presents theoretical results that show that the cost of performing six specific incremental graph computations, e.g., regular path queries and strongly connected components, cannot be bounded by only the size of the changes $\Delta G$ in input and output. Then, they develop and evaluate algorithms that have guarantees in a more relaxed notion of boundedness (based on the subset of the graph inspected by the batch algorithm being incrementalized). 
Similar to the incremental PR algorithm of GB, these are incremental algorithms specialized to specific algorithms  and can be more efficient than differential computation. 

This is both a limitation and an advantage of our design because we take the burden of designing incremental algorithms
but share computation through a black-box technique.
For expert users, we could allow users to program specialized incremental algorithms by exposing a TD-based programming interface. Using this interface we can allow programmers to implement specialized incremental functions %
and execute them directly on TD, which can be more efficient than using DD.

\textbf{Tegra~\cite{iyer:tegra}} is a system developed on top of Apache Spark~\cite{zaharia:spark}, that is designed to perform ad-hoc window-based analytics on a dynamic graph. 
Specifically, Tegra allows users to tag arbitrary snapshots of their graphs with timestamps. 
The system has a technique for sharing arbitrary computation across snapshots through a differential computation-like computation maintenance logic. However, the system is  optimized for retrieving arbitrary snapshots quickly instead of sharing computation across snapshots efficiently.
Similar to GB, we can also develop a \graphsurge{}-like system on top of Tegra, however
The architectures of DD and Tegra are more similar than DD and GB, as both are dataflow systems but the authors did not provide us with the code to compare. 
a performance comparison between Tegra and DD from reference~\cite{iyer:tegra} reports Tegra's performance to be significantly slower than DD (Figure~14) for incrementally maintained computations.

\customsection{Kaskade~\cite{trindade:kaskade}} is a graph query optimization framework from Microsoft that uses {\em materialized graph views} to speed up query evaluation. Specifically, 
Kaskade takes as input a query workload $\mathcal{Q}$ and an input graph $G$. Then, the framework enumerates possible {\em views} for $\mathcal{Q}$, which are other graphs $G'$ that contain a subset of the nodes in $G$ and other edges that can represent multi-hop connections in $G$. 
Kaskade selects a set of these views, materializes them in Neo4j, and then translates queries in $\mathcal{Q}$ to use appropriate views. The system is not designed for the applications that \graphsurge{} is designed to support, where users define a set of views, that are snapshots of a graph $G$ and run analytics computations on multiple views.

\customsection{GraphGen~\cite{xirogiannopoulos:graphgen}} is a system to extract graphs out of relational tables stored in an RDBMS. Each graph is a relational view that describes a nodes table and another view that describes an edge table. Users can extract many views that are sometimes stored in memory and sometimes in the RDBMS. The focus of this work is addressing how to store  very large extracted graphs in compressed format.
The system is not designed to define multiple graphs and share computations across them.  

\customsection{Temporal and Streaming Graph Analytics Systems:}

\customsection{SAMS~\cite{then:sams}} is a system to execute a single algorithm on multiple snapshots. However, the system does not have any computation sharing capabilities similar to DD. Running a WCC algorithm on $k$ views would result in rerunning WCC $k$ times from scratch. Instead, the system is optimized for sharing accesses to the same parts of the graph to increase data locality across these $k$ computations. 

\customsection{Delta-Graph~\cite{khurana:deltagraph}} is a system designed for temporal
analysis but the system is designed primarily for retrieval of views of a
dynamic graph in arbitrary timestamps and not for performing analytics. 

\customsection{Gelly Streaming~\cite{gelly-streaming}} is a library on top of Flink~\cite{carbone:flink} to program pure streaming computations with a graph API. Users have to implement their own streaming, computation maintenance, and operator synchronization logic as a stream of edges arrives at Flink. Similar to the temporal graph analytics systems, it does not provide any general computation sharing capabilities. Therefore, it is not appropriate to use as an execution layer for \graphsurge{}, as it would require developing a DD-like layer on top of it.

\vspace{-5pt}
\section{Conclusions and Future Work}\label{sec:conclusion}
We presented the design and implementation of \graphsurge{}, an open-source view-based graph analytics system, developed on top of the TD system and its DD layer. \graphsurge{} allows users to define arbitrary views over their graphs, organize these views into view collections, and perform arbitrary graph analytics using a DD-based analytics API\@. \graphsurge{} is motivated by real-world applications, such as perturbation analysis or analysis of the evolution of large-scale networks, that require capabilities to analyze multiple, sometimes hundreds of views of static input graphs efficiently. We presented two optimization problems, the collection ordering and splitting problems, for which we described efficient algorithms and studied the performances of our optimizations.
\graphsurge{}'s approach for computation sharing is based on differential computation.
As future work, we are interested in studying modifications one can make to the internals of DD to share computations more efficiently, for example using techniques from
incremental versions of specific graph algorithms~\cite{ammar:bigjoin, mariappan:graphbolt}.%

\section{Acknowledgments}
We are grateful to Nafisa Anzum, Pranjal Gupta, and Xiyang Feng for their help at various stages of this paper, and to Frank Mcsherry for answering queries related to Differential Dataflow. We thank the anonymous reviewers for their valuable comments. This research was funded by a grant from Waterloo-Huawei Joint Innovation Laboratory.

\bibliographystyle{ACM-Reference-Format}
\bibliography{references}

\end{document}